\documentclass[a4paper,unpublished,twocolumn,11pt]{quantumarticle}
\pdfoutput=1

\usepackage[numbers, sort&compress]{natbib}
\usepackage[utf8]{inputenc}
\usepackage[english]{babel}
\usepackage[T1]{fontenc}
\usepackage{bookmark}
\usepackage{hyperref}
\hypersetup{
    colorlinks=true,
    linkcolor=blue,
    filecolor=magenta,      
    urlcolor=cyan
    }
\usepackage{enumitem}

\usepackage{amsmath,amssymb,amsthm,bm,amsfonts,mathrsfs,bbm,mathtools}
\usepackage{physics}
\newcommand{\Qover}{\mathcal{Q}_\delta}

\usepackage{algorithmicx}
\usepackage[Algorithm,ruled]{algorithm}
\usepackage{algcompatible}

\newtheorem{theorem}{Theorem}
\newtheorem{definition}{Definition}


\usepackage{tikz}
\usetikzlibrary{arrows.meta, positioning, shapes.geometric, fit, calc}
\usepackage{pgfplots}
\tikzset{
  arrow/.style={thick, -{Latex[length=3mm,width=2mm]}}
}
\tikzset{shifted path/.style args={from #1 to #2 by #3}{insert path={
let \p1=($(#1.east)-(#1.center)$),
\p2=($(#2.east)-(#2.center)$),\p3=($(#1.center)-(#2.center)$),
\n1={veclen(\x1,\y1)},\n2={veclen(\x2,\y2)},\n3={atan2(\y3,\x3)} in
(#1.{\n3+180+asin(#3/\n1)}) to 
(#2.{\n3-asin(#3/\n2)})
}}}
\usepackage{tikz-3dplot}
\usepackage{soul,xcolor}
\definecolor{LightBlue}{HTML}{3091DD}
\definecolor{LightGreen}{HTML}{32C5CB}

\usepackage[final]{changes} 
\definechangesauthor[name={Gabi}, color=red]{gab}
\definechangesauthor[name={Tom}, color=green]{tom}
\definechangesauthor[name={Pablo}, color=blue]{pab}

\pgfplotsset{compat=1.17}
\begin{document}

\title{Deterministic randomness extraction for quantum random number generation with partial trust}

\author{Pablo Tikas Pueyo}
\orcid{0009-0005-2835-0760}
\affiliation{Quside Technologies S.L., C/Esteve Terradas 1, 08860 Castelldefels, Barcelona, Spain}
\affiliation{Universitat de Barcelona, C/de Martí i Franquès, 1, 11, Distrito de Les Corts, 08028 Barcelona, Spain}

\author{Tomás Fernández Martos}
\email{tfernandez@quside.com}
\orcid{0009-0006-8697-3541}
\affiliation{Quside Technologies S.L., C/Esteve Terradas 1, 08860 Castelldefels, Barcelona, Spain}
\affiliation{ICFO-Institut de Ciències Fotòniques, The Barcelona Institute of Science and Technology, 08860 Castelldefels, Barcelona,
Spain}

\author{Gabriel Senno}
\orcid{0000-0002-4788-1664}
\affiliation{Quside Technologies S.L., C/Esteve Terradas 1, 08860 Castelldefels, Barcelona, Spain}
\affiliation{ICFO-Institut de Ciències Fotòniques, The Barcelona Institute of Science and Technology, 08860 Castelldefels, Barcelona,
Spain}

\maketitle

\begin{abstract}
It is a well-known fact in classical information theory that no deterministic procedure can extract close-to-ideal randomness from an arbitrary entropy source. On the other hand, if additional knowledge about the source is available---e.g., that it is a sequence of independent Bernoulli trials---then deterministic extractors do exist. For quantum entropy sources, where in addition to classical random variables we consider quantum side information, the use of extra knowledge about their structure was pioneered in a recent publication [C. Foreman and L. Masanes, Quantum 9, 1654 (2025)]. In that work, the authors provide deterministic extractors for device-independent randomness generation with memoryless devices achieving a sufficiently high CHSH score. In this work, we port their construction to the prepare-and-measure scenario. Specifically, we prove that the considered functions are also extractors for memoryless devices in settings with partial trust, either in the state preparation or in the measurement, as well as in a semi-device-independent setting under an overlap assumption on the prepared quantum states. Within this last setting, we simulate the resulting randomness generation protocol on a novel and experimentally relevant family of behaviors, observing positive key rates already for $7\times 10^3$ rounds.
\end{abstract}

\section{Introduction}

Randomness has become an indispensable asset in the digital era, underpinning a wide range of applications: from computer simulations and forecasting to artificial intelligence and, most critically, cryptography~\cite{hayes2001computing}. Recent technological advances have enabled the emergence of commercially available quantum random number generators (QRNGs), exploiting the inherent unpredictability of quantum mechanics~\cite{herrero2017quantum}.

The general structure of a QRNG consists of two main components: an \textit{entropy source} and a \textit{randomness extractor}. The entropy source is in charge of producing \emph{raw} bit strings with some amount of private randomness, usually quantified via the (smooth) conditional min-entropy~\cite{konig2009operational} (see also~\cite{dupuis2023privacy}). The randomness extractor then transforms the raw bit string into another, usually shorter, bit string that is $\epsilon$-close (in trace distance) to being perfectly random, i.e., uniformly random and independent of a potential eavesdropper's side information.

It is straightforward to see that no deterministic procedure can extract randomness from the class of all entropy sources with a given min-entropy lower bound. In a celebrated result, Santha and Vazirani showed that this impossibility continues to hold even for a restricted, yet natural, class of sources with arbitrary high min-entropy~\cite{715945}. Faced with this limitation, two main approaches were developed in the literature. In \textit{seeded extraction}, a short string of perfectly random bits is used as a catalyst to extract randomness from weak sources. In \textit{multi-source extraction}, additional weak sources are combined to extract randomness under the assumption that each contains sufficient conditional entropy, and some form of independence between them and the side-information also holds (see~\cite{arnonfriedman_et_al:LIPIcs.TQC.2016.2} and, more recently,~\cite{sandfuchs2025randomness}).

Deterministic extractors, nevertheless, do exist for particular classes of entropy sources. For example, biased yet independent coins can be XORed to produce a close to uniform bit (see, e.g.,~\cite{vadhan2012pseudorandomness}). Naturally, this fact led to the question of whether one could construct deterministic extractors tailored for particular families of  quantum entropy sources. In a remarkable result~\cite{Foreman_2025}, the first quantum random number generation protocol with deterministic extraction was constructed. Focusing on a device-independent (DI) setting, the authors prove that the violation of a Bell inequality can replace a min-entropy lower bound as the extractor's promise.

In this work, we show that the techniques developed in~\cite{Foreman_2025} for the DI setting can be adapted to multiple prepare-and-measure settings with \emph{partial trust}, either in the state generation or in the measurement. We also extend the results to the semi-DI framework based on an overlap assumption on the prepared states developed in~\cite{Brask_2017}. Our main result shows that the role played by a Bell functional in the DI construction can be played by a solution to the dual of the guessing probability in any of these partial trust settings. Focusing on the semi-DI setting, our results together with the techniques from~\cite{Foreman_2025} lead to a spot-checking randomness generation protocol whose performance we simulate with a quantum entropy source based on an experimentally relevant family of behaviors~\cite{mir2025measurement}. In the finite-size regime, we observe that positive key rates can be achieved for as low as $7\times 10^3$ rounds. Whereas in the asymptotic limit of infinitely many rounds, we show that even if the fraction of estimating rounds is vanishing, successful extraction can be achieved from arbitrarily close to deterministic behaviors.

The paper is structured as follows. In Sec.~\ref{sec:assumptions}, we describe the different prepare-and-measure settings we work with. In Sec.~\ref{sec:main-result}, we state our main theorem: a family of operator inequalities resulting from dual solutions to the guessing probability. In Sec.~\ref{sec:extractors}, we use this result to prove that the functions considered in~\cite{Foreman_2025} are also extractors in our setting. In Sec.~\ref{sec:protocol}, we provide a randomness generation protocol focused on the semi-DI setting, based on the one from~\cite{Foreman_2025}, whose expected rates under honest but noisy implementations we simulate in Sec.~\ref{sec:numerics}. Finally, we conclude with some remarks in Sec.~\ref{sec:conclusions}.

\section{Settings and assumptions}\label{sec:assumptions}

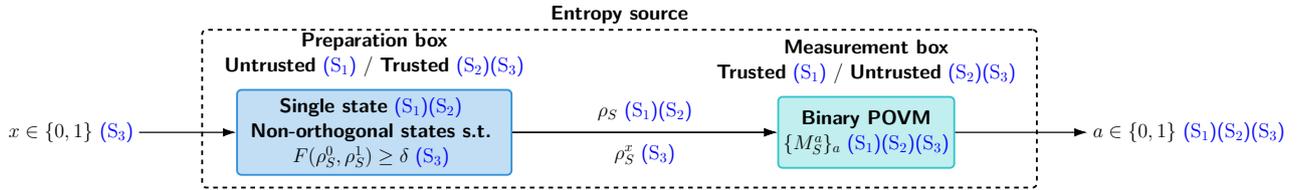
\begin{figure*}
\centering
\resizebox{\textwidth}{!}{
\begin{tikzpicture}[font=\Large\sffamily, node distance=1.5cm and 2.5cm, every node/.style={scale=0.9}]
  \tikzstyle{prepbox} = [rectangle, rounded corners = 1mm, draw=LightBlue, very thick, fill=LightBlue!30, minimum width=3.2cm, minimum height=1.8cm, align=center]
  \tikzstyle{measbox} = [rectangle, rounded corners = 1mm, draw=LightGreen, very thick, fill=LightGreen!30, minimum width=3.2cm, minimum height=1.8cm, text=black, align=center]
  \tikzstyle{outerbox} = [draw, rounded corners = 1mm, dashed, very thick,minimum width=21cm, minimum height=4.0cm, rectangle]
  \tikzstyle{labelnode} = [font=\Large\sffamily\bfseries]
  \tikzstyle{arrow} = [thick, -{Latex[length=3mm,width=2mm]}]

  \node[prepbox] (prep) {\Large{\begin{tabular}{c} \textbf{Single state}~\ref{set:untrusted_prep}\ref{set:untrusted_meas} \\ \textbf{Non-orthogonal states s.t.} \\ \Large$F(\rho_S^0,\rho_S^1) \geq \delta$~\ref{set:SDI} \end{tabular}
  }};
  \node[measbox, right=6.0cm of prep] (meas) {\Large \textbf{Binary POVM} \\\Large$\{M^a_S\}_{a}$~\ref{set:untrusted_prep}\ref{set:untrusted_meas}\ref{set:SDI} };

  \node[left=2.2cm of prep] (x) {\Large $x \in \{0,1\}$~\ref{set:SDI}};
  \node[right=3.0cm of meas] (a) {\Large $a \in \{0,1\}$~\ref{set:untrusted_prep}\ref{set:untrusted_meas}\ref{set:SDI}};

  \draw[arrow] (x) -- (prep);
  \draw[arrow] (prep) -- (meas) node[midway, above, yshift=0.15cm] {\Large $\rho_S$~\ref{set:untrusted_prep}\ref{set:untrusted_meas}};
  \draw[arrow] (prep) -- (meas) node[midway, below, yshift=-0.15cm] {\Large $\rho_S^x$~\ref{set:SDI}};
  \draw[arrow] (meas) -- (a);

  \node[above=0.1cm of prep] {\Large{\begin{tabular}{c} \textbf{Preparation box} \\ \textbf{Untrusted}~\ref{set:untrusted_prep} / \textbf{Trusted}~\ref{set:untrusted_meas}\ref{set:SDI} \end{tabular}
  }};
  \node[above=0.1cm of meas] {\Large{\begin{tabular}{c} \textbf{Measurement box} \\ \textbf{Trusted}~\ref{set:untrusted_prep} / \textbf{Untrusted}~\ref{set:untrusted_meas}\ref{set:SDI} \end{tabular}
  }};

  \node[outerbox, fit=(prep)(meas), yshift=0.6cm, xshift = 0.6cm, label=above:{\Large \textbf{Entropy source}}] {};

\end{tikzpicture}%
}
\caption{\textbf{Setup scheme.}~Schematic representation of the partial trust model considered for the entropy source featuring Sett.~\ref{set:untrusted_prep}-\ref{set:SDI}, including trust in different parts of the device, as well as different possible levels of characterization.}
\label{fig:EntropySource-Model}
\end{figure*}

We consider a \textit{prepare-and-measure} (PM) scenario with an honest user (Alice) holding a device $\mathcal{D}$ which, depending on the setting, features:
\begin{enumerate}[label=($\textrm{S}_{\arabic*}$)]
    \item \textbf{Untrusted-preparation}. A (characterized or not) untrusted preparation $\rho_S$ and a trusted and characterized binary-outcome measurement described by a \textit{positive operator-valued measure} (POVM) $\{M^a_S\}_{a\in\{0,1\}}$.\label{set:untrusted_prep}
    
    \item \textbf{Untrusted-measurement}. A trusted and characterized preparation $\rho_S$ and an untrusted but characterized\footnote{The case of an untrusted and uncharacterized binary-outcome measurement is not included because the only result is the trivial one, where the guessing probability for an eavesdropper is always equal to 1.} binary-outcome measurement described by a POVM $\{M^a_S\}_{a\in\{0,1\}}$.\label{set:untrusted_meas}
    
    \item \textbf{Semi-DI~\cite{Brask_2017, Flatt_2022}}. A pair of trusted preparations $\{\rho_S^x\}_{x\in\{0,1\}}$ with a known lower bound on their fidelity $$F(\rho_S^0,\rho_S^1) := \Tr\left[\sqrt{\sqrt{\rho_S^0}\rho_S^1\sqrt{\rho_S^0}}\right]^2\geq \delta$$ and a (characterized or not) untrusted binary-outcome measurement described by a POVM $\{M^a_S\}_{a\in\{0,1\}}$.\label{set:SDI}
\end{enumerate}
When we say that the preparations or the measurements are \textit{characterized}, we mean that an explicit classical description of the corresponding quantum states or POVMs is assumed. 

An eavesdropper, Eve, sitting outside Alice's secure location holds a quantum system $E$ potentially correlated with her device. We say that part of Alice's device is \textit{untrusted} when it is allowed to be correlated with Eve's system, i.e., when Eve is assumed to hold quantum side information about it. In Setts.~\ref{set:untrusted_meas} and~\ref{set:SDI} we trust the preparations, while in Sett.~\ref{set:untrusted_prep} it is the measurement that we assume to be trusted.

Apart from the characterization and trust considerations, we also include the following \textit{single-round} assumptions, which apply to all settings unless explicitly stated:
\begin{enumerate}[label=($\textrm{A}_{\arabic*}$)]
    \item \textbf{Uncorrelated PM boxes.} There are no correlations between the preparations and the measurement device.\label{ass:uncorr-prep-and-meas}
        
    \item \textbf{No leakage.} While evident from the context, it is assumed that Eve does not have direct access to the measurement outcomes.\label{ass:no-leakage}

    \item \textbf{Input randomness}~\ref{set:SDI}\textbf{.} The inputs $x\in\{0,1\}$ to the preparation box can be chosen uniformly at random and independently of the device $\mathcal{D}$ and of Eve's side information.\label{ass:input-randomness}
\end{enumerate}

The setup, together with the assumptions we make tailored for each setting, is schematized in Fig.~\ref{fig:EntropySource-Model}. We will denote with $\mathcal{Q}_\mathrm{M}$, $\mathcal{Q}_\mathrm{P}$, and $\Qover$ the sets of \emph{behaviors} $\mathbf{p}$ achievable within Setts.~\ref{set:untrusted_prep},~\ref{set:untrusted_meas}, and~\ref{set:SDI}, respectively. That is, the sets of probability distributions $\{p(a)\}_a$ such that there exist $\rho_S$ and $\{M^a_S\}_a$ with $p(a) = \Tr[M^a_S\rho_S]$ for Setts.~\ref{set:untrusted_prep} and~\ref{set:untrusted_meas}, or the set of conditional distributions $\{p(a|x)\}_{a,x}$ such that there exist $\{\rho_S^x\}_x$ and $\{M^a_S\}_a$ with $p(a|x)=\Tr[M^a_S\rho_S^x]$ and $F(\rho_S^0,\rho_S^1)\geq \delta$ for Sett.~\ref{set:SDI}.

In order to be able to refer to the correlations between the outcomes of Alice's device $\mathcal{D}$ and Eve's \emph{quantum} side-information $E$ for the untrusted measurement Setts.~\ref{set:untrusted_meas} and~\ref{set:SDI}, we introduce the \emph{generalized Naimark dilation} picture of the measurement device~\cite{frauchiger2013true}. Namely, we use the fact that a POVM $\{M^a_S\}_a$ can always be seen as projective measurement $\{\Pi_{SM}^a\}_a$ on the original system $S$ and an ancillary system $M$. When the measurement is untrusted, the correlations with Eve are modelled by letting $M$ be in some mixed state $\sigma_M$ and giving Eve its purification~\cite{frauchiger2013true,Senno_2023}. In this picture, the above assumptions thus imply that the global state of systems $S$, $M$ and $E$  is of the form 
\begin{align*}
    \rho_{SME}&=\dyad{\Psi}_{SE}\otimes \dyad{0}_M \text{ \ref{set:untrusted_prep}},\\
    \rho_{SME}&=\rho_{S}\otimes \dyad{\Phi}_{ME} \text{ \ref{set:untrusted_meas}},\\
    \rho_{SME}^x&=\rho_{S}^{x}\otimes \dyad{\Phi}_{ME} \text{ on input $x$~\ref{set:SDI}}.
\end{align*}

After $n_r$ uses of Alice's device $\mathcal{D}$, a sequence $\mathbf{a}\in \{0,1\}^{n_r}$ of \emph{raw} random bits is produced, with corresponding inputs $\mathbf{x}\in \{0,1\}^{n_r}$ in the case of Sett.~\ref{set:SDI}. We will make the following additional assumptions regarding these \emph{sequential} uses of $\mathcal{D}$:

\begin{enumerate}[resume*]
    \item \textbf{Memoryless measurements.} Each measurement acts only on the quantum system of the current round, acting trivially on all others. Formally, the measurement operators across $n_r$ rounds are of the form 
    \begin{equation*}
        \Pi_{S_rM_r}^{\mathbf{a}}=\bigotimes_{i=1}^{n_r}\Pi_{S_iM_i}^{a_i}.
    \end{equation*}
    This memoryless assumption comes from the construction in~\cite{Foreman_2025}, on which our results are based. In Setts.~\ref{set:untrusted_meas} and~\ref{set:SDI}, the $n_r$-round state $\sigma_{M_r}$ can be arbitrarily correlated across rounds.\label{ass:memoryless}

    \item \textbf{Independent preparations}~\ref{set:untrusted_meas}\ref{set:SDI}\textbf{.} The state prepared in round $i$ is independent of the states prepared in previous rounds $j<i$. Formally, the $n_r$-rounds preparations are of the form
    \begin{equation*}
        \rho_{S_r} = \bigotimes_{i=1}^{n_r} \rho_{S_i}\text{~\ref{set:untrusted_meas}}, \qquad \rho_{S_r}^\mathbf{x} = \bigotimes_{i=1}^{n_r} \rho_{S_i}^{x_i}\text{~\ref{set:SDI}}.
    \end{equation*}\label{ass:indep-prep}
\end{enumerate}

The correlations between the raw random bits $\mathbf{a}$ and Eve's side-information $E$ are described by the classical-quantum (cq) state
\begin{align}
    &\qquad\qquad\rho_{\mathbf{A}E}=\sum_{\mathbf{a}\in\{0,1\}^{n_r}}\dyad{\mathbf{a}}\otimes \tilde{\rho}_{E}^{\mathbf{a}}, &&\nonumber\\
    &\text{with}&&\nonumber\\
    &\qquad\tilde{\rho}_{E}^\mathbf{a}=\Tr_{S_rM_r}[(\Pi_{S_rM_r}^\mathbf{a}\otimes\mathbbm{1}_E)\rho_{S_rM_rE}],&&\nonumber\\
    &\text{where}&&\nonumber\\
    &\qquad\rho_{S_rM_rE} =\dyad{\Psi}_{S_rE}\otimes\dyad{0}_{M}^{\otimes n_r}\text{~\ref{set:untrusted_prep},}&&\nonumber\\
    &\qquad\rho_{S_rM_rE} =\rho_{S_r}\otimes\dyad{\Phi}_{M_rE}\text{~\ref{set:untrusted_meas},}&&\nonumber\\
    &\qquad\rho_{S_rM_rE} =\rho_{S_r}^\mathbf{x}\otimes\dyad{\Phi}_{M_rE}\text{~\ref{set:SDI}.}&&\label{eq:cq-state}
\end{align}

\section{Operator inequalities from dual solutions to the guessing probability}\label{sec:main-result}
The seedless extraction protocol from~\cite{Foreman_2025} rests on an inequality from~\cite{masanes2011secure} that bounds the \emph{predictabilty} (see below) of the outcomes in a CHSH experiment by a function of the observed Bell violation. In this section, we prove analogous results in Setts.~\ref{set:untrusted_prep}-\ref{set:SDI}. Namely, the role played by the (shifted) CHSH functional in~\cite{Foreman_2025} will be played by a solution to the dual of the semidefinite program (SDP)~\cite{Skrzypczyk_2023} defining Eve's \emph{guessing probability} in each of the considered settings. For the sake of clarity, we focus on Sett.~\ref{set:SDI} for our discussion since it is the most experimentally relevant for the context of our work, and only present the main results for the other settings, with the corresponding derivations included in the appendices.

Let us consider a single use of Alice's entropy source $\mathcal{D}$ within Sett.~\ref{set:SDI}. Under Assumptions~\ref{ass:uncorr-prep-and-meas}-\ref{ass:input-randomness}, the probability that Eve guesses $\mathcal{D}$'s outcome when the input is $x=0$, given some observed behavior $\mathbf{p}$ and a $\delta$-overlap assumption, is defined by the following SDP~\cite{Brask_2017}
\begin{align}\label{eq:classical-pguess-povm}
    & P_{{\rm guess}}(A|E,X=0,\mathbf{p},\delta):=&&\nonumber\\ 
    &\max_{\{p(\lambda)\}_{\lambda\in\{0,1\}},\{M_{a|\lambda}\}_{a,\lambda}}\sum_{a} p(a)\bra{\psi^0}M_{a|a}\ket{\psi^0}, &&\nonumber\\
    &\qquad\quad~\textrm{subject to } &\nonumber\\
    &\qquad\qquad\quad M_{a|\lambda}\succeq0,~\sum_a M_{a|\lambda}=\mathbbm{1}, &\nonumber\\
    &\qquad\qquad\quad\sum_\lambda p(\lambda) \bra{\psi^x}M_{a|\lambda}\ket{\psi^x} = p(a|x),
\end{align}
where, without loss of generality, $\ket{\psi^0}=\ket{0}$ and $\ket{\psi^1}=\sqrt{\delta}\ket{0}+\sqrt{1-\delta}\ket{1}$. Note that in Eq.~\eqref{eq:classical-pguess-povm} we are considering the measurement to be uncharacterized, a specific case of Sett.~\ref{set:SDI} which is less restrictive and upper bounds the scenario with a characterized measurement (in terms of guessing probability). The classical noise variable $\lambda$ represents Eve's side information and, for pure states and noisy measurements, this formulation of $P_{\mathrm{guess}}$ has been shown to be equivalent to the case where Eve holds quantum side information~\cite{Senno_2023}. 

In App.~\ref{app: proof of bound on p_g}, we derive the \emph{dual} of this SDP, given by
\begin{align}\label{eq:dual}
    & P_{{\rm guess}}^*(A|E,X=0,\mathbf{p},\delta):= &&\nonumber\\
    &\min_{\{H_\lambda\}_{\lambda\in\{0,1\}}\subseteq\mathbb{C}^{2\times 2},\{\nu_{a,x}\}_{a,x}\subseteq \mathbb{R}} \sum_{a,x=0}^1 \nu_{a,x}p(a|x), &&\nonumber\\
    &\qquad\quad~\textrm{subject to } &\nonumber\\
    &\qquad\qquad\quad(H_\lambda)^\dagger=H_\lambda,\nonumber\\
    &\qquad\qquad\quad\dyad{\psi^0}\delta_{a,\lambda} - \sum_{x=0}^1 \dyad{\psi^x} \nu_{a,x} &&\nonumber\\
    &\qquad\qquad\qquad\qquad\quad +H_\lambda - \frac{1}{2}\Tr\Big[H_\lambda\Big]\mathbbm{1} \preceq 0,
\end{align}
where $\delta_{a,\lambda}$ is the Kronecker delta. For every behavior $\mathbf{p}$ and every $\delta\in[0,1]$, any feasible solution $\langle \{H_\lambda\},\{\nu_{a,x}\} \rangle $ to Eq.~\eqref{eq:dual} provides an upper bound 
\begin{align}\label{eq:dual-upper-bounds-probs}
P_{{\rm guess}}(A|E,X=0,\mathbf{p},\delta)\leq \sum_{a,x=0}^1 \nu_{a,x}p(a|x)
\end{align}
to the solution of the \emph{primal} in Eq.~\eqref{eq:classical-pguess-povm}. In App.~\ref{app: extra sdps}, we provide the formulations of the primal and dual SDPs for Eve's guessing probability in Setts.~\ref{set:untrusted_prep} and~\ref{set:untrusted_meas}, which lead to bounds that are similar in structure to the bound in Eq.~\eqref{eq:dual-upper-bounds-probs}.

Since $P_{{\rm guess}}(A|E,X=0,\mathbf{p},\delta)$ is always greater than or equal to the largest probability $\max_a{p(a|0)}$ (Eve can always choose to ignore her side-information), from Eq.~\eqref{eq:dual-upper-bounds-probs} it follows that
\begin{align*}
    p(a|0)
    &\leq P_{\mathrm{guess}}(A|E,X=0,\mathbf{p},\delta) \\
    &\leq \sum_{a,x=0}^1 \nu_{a,x} p(a|x).
\end{align*}
Therefore, for every $\mathbf{p}\in\Qover$, we have that 
\begin{equation}\label{eq: predictability bound}
    \Big| p(0|0) - p(1|0) \Big| \leq 2\sum_{a,x=0}^1 \nu_{a,x} p(a|x) - 1.
\end{equation}
Following \cite{Foreman_2025}, we refer to the lhs of Eq.~\eqref{eq: predictability bound} as the outcome's \emph{predictability} for the input $x$. Our main technical contribution, Thm. \ref{thm: ancilla operator inequality} below (and the analogous Thms. for the other settings), states an operator version of the inequality in Eq.~\eqref{eq: predictability bound}. Crucially, however, by switching to the Naimark dilation picture, the operators in this inequality act on the Hilbert space $\mathcal{H}_{M}$ of the ancillary system $M$ within the measurement device, which, in this setting, is precisely the system to which Eve’s system is correlated.
\begin{theorem}\label{thm: ancilla operator inequality}
    For every $\delta\in[0,1]$, if $\langle \{\nu_{a,x}\}_{a,x=0}^1,\{H_\lambda\}_{\lambda=0}^1\rangle$ is a solution to Eq.~\eqref{eq:dual} for some behavior $\mathbf{p}\in\Qover$, then, for any $\{\rho_{S}^x\}_{x=0}^1$ such that $F(\rho_S^0,\rho_S^1)\geq\delta$ and for any $\{\Pi_{SM}^a\}_{a=0}^1$, the following inequality over operators in $\mathcal{H}_M$ holds:
    \begin{equation}\label{eq: ancilla operator inequality}
        \pm \Tr_S\Big[(\rho_{S}^0 \otimes \mathbbm{1}_{M})(\Pi_{SM}^0-\Pi_{SM}^1)\Big] \preceq G_{M}, 
    \end{equation}
    with $G_M \coloneqq 2\sum_{a,x} \nu_{a,x} \Tr_{S}\Big[(\rho_{S}^x \otimes \mathbbm{1}_{M}) \Pi_{SM}^a\Big] - \mathbbm{1}_{M}$.
\end{theorem}

The proof of this theorem is provided in App.~\ref{app: proof of ancilla operator inequality}. Furthermore, as we show in App.~\ref{app: extension of ancilla operator inequality}, Thm.~\ref{thm: ancilla operator inequality} can also be formulated in terms of Setts.~\ref{set:untrusted_prep} and~\ref{set:untrusted_meas}. 

In Sett.~\ref{set:untrusted_meas}, the operator inequality acts on the Hilbert space $\mathcal{H}_M$, similarly to Eq.~\eqref{eq: ancilla operator inequality}, and reads
\begin{equation}\label{eq: sett1 ancilla operator inequality}
    \pm \Tr_S\Big[(\rho_{S} \otimes \mathbbm{1}_M)(\Pi_{SM}^0-\Pi_{SM}^1)\Big] \preceq G_{M}, 
\end{equation}
with $G_M \coloneqq 2\sum_{a,k} \nu_{a,k} \Tr_{S}\Big[(\varphi^k_{S} \otimes \mathbbm{1}_{M}) \Pi_{SM}^a\Big] - \mathbbm{1}_{M}$, where $\rho_S$ is a trusted and characterized state preparation, and $\{\varphi_S^k\}_k$ is a tomographically complete set of states whose action unequivocally defines the untrusted measurement $\{M^a_S\}_a$. 

In Sett.~\ref{set:untrusted_prep}, the inequality acts instead on the Hilbert space $\mathcal{H}_S$ to which Eve is correlated since the measurement is now trusted, and it reads
\begin{equation}\label{eq: sett2 ancilla operator inequality}
    \pm \Tr_M\Big[(\mathbbm{1}_S \otimes \dyad{0}_M)(\Pi_{SM}^0-\Pi_{SM}^1)\Big] \preceq G_{S}, 
\end{equation}
with $G_S \coloneqq 2\sum_{b} \nu_{b} \Tr_{M}\Big[(\mathbbm{1}_S \otimes \dyad{0}_{M}) \tilde{\Pi}_{SM}^b\Big] - \mathbbm{1}_{S}$, where $(\dyad{0}_M, \{\Pi_{SM}^a\}_a)$ and $(\dyad{0}_M, \{\tilde{\Pi}_{SM}^b\}_{b\leq k})$ are, respectively, (standard) Naimark dilations of the trusted and characterized measurement $\{M^a_S\}_a$ and of some (any) $k$-outcome informationally complete (IC) POVM on $\mathcal{H}_S$  (whose statistics over $\rho_S$ completely characterize it).

In the following section, we use Thm.~\ref{thm: ancilla operator inequality} to apply the extractors considered in \cite{Foreman_2025} to our prepare-and-measure Setts.~\ref{set:untrusted_prep}-\ref{set:SDI}. In the section after that, focusing specifically on Sett.~\ref{set:SDI}, we derive a spot-checking randomness generation protocol that uses these extractors.

\section{Single- and multi-bit deterministic extractors}\label{sec:extractors}
In a QRNG, the randomness extractor component is in charge of transforming the raw random bits produced by some intrinsically random quantum measurement (the quantum entropy source) into another, usually shorter, string satisfying the universally composable notion of security given by the trace distance \cite{frauchiger2013true,renner2005universally}. In this work, we focus on \emph{deterministic randomness extractors}. 
\begin{definition}
We say that a function $g:\{0,1\}^{n_r}\to\{0,1\}^m$ is a \emph{$(m,\epsilon)$-deterministic randomness extractor} for a cq-state $\rho_{\mathbf{A}E}$ iff the cq-state \begin{equation}\label{eq: actual cq-state}
    \rho_{\mathbf{K}E} = \sum_{\mathbf{k}\in\{0,1\}^m}\sum_{\mathbf{a}\in\{0,1\}^{n_r} : g(\mathbf{a})=\mathbf{k}} \dyad{\mathbf{k}} \otimes \tilde{\rho}_{E}^\mathbf{a},
\end{equation}
is such that
\begin{equation}\label{eq: trace distance}
    \frac{1}{2}\left\Vert \rho_{\mathbf{K}E} - \rho_{\mathbf{K}E}^{\text{ideal}} \right\rVert_1 \leq \epsilon, 
\end{equation}
with $\rho_{\mathbf{K}E}^{\text{ideal}} = \frac{\mathbbm{1}}{2^m} \otimes \rho_E$.
\end{definition}
In what follows, we show that the functions shown to be deterministic extractors for device-independent quantum entropy sources in \cite[Theorems 2 and 4]{Foreman_2025} are also extractors for the outcomes of $n$ sequential uses of the device described in Sec.~\ref{sec:assumptions}.

\paragraph{Single-bit extraction.} The single-bit extractor proposed in~\cite{Foreman_2025} is a well-known function in computer science: the \texttt{XOR} function. Applying it to the measurement outcomes $\mathbf{a}$ yields
\begin{equation}\label{eq:XOR_function}
    \texttt{XOR}(\mathbf{a}) = \bigoplus_{i=1}^{n_r} a_i \in \{0,1\}.
\end{equation}

\begin{theorem}\label{thm: single-bit extraction security proof}
    Let $\rho_{\mathbf{A}E}$ be the cq-state that results from $n_r$ sequential uses of a PM device satisfying assumptions \ref{ass:uncorr-prep-and-meas}-\ref{ass:indep-prep} as described by Eq. \eqref{eq:cq-state}, for a given overlap bound $\delta\in[0,1]$ in the case of Sett.~\ref{set:SDI}. Then, for all $n_r>0$, the function $\texttt{XOR}:\{0,1\}^{n_r}\to\{0,1\}$ is a $(1,\epsilon_n)$-deterministic randomness extractor for $\rho_{\mathbf{A}E}$ with
    \begin{equation*}
        \epsilon_n = \frac{1}{2}\Tr\Bigg[\phi\prod_{i=1}^{n_r}G_{i}\Bigg],
    \end{equation*}
    where $\phi=\rho_{S_r}$ for Sett.~\ref{set:untrusted_prep} and $\phi=\sigma_{M_r}$ for Setts.~\ref{set:untrusted_meas} and \ref{set:SDI}, and $G_{i}$ is the per-round version of the operators $G$ introduced in Thm.~\ref{thm: ancilla operator inequality}, defined in Eqs.~\eqref{eq: ancilla operator inequality}-\eqref{eq: sett2 ancilla operator inequality} for Setts.~\ref{set:SDI}-\ref{set:untrusted_prep}, respectively.
\end{theorem}
The proof of Thm. \ref{thm: single-bit extraction security proof} is given in App.~\ref{app: proof of single-bit extraction security proof}. To gain some intuition about this result, let us focus again on Sett.~\ref{set:SDI} as we did in Sec.~\ref{sec:main-result} and consider a simplification under an \textit{independent and identically distributed} (IID) assumption. In other words, let us consider that $n_r$ sequential uses of Alice's device $\mathcal{D}$ consist on $n_r$ independent uses of the same \textit{single-round} device $\mathcal{D}_{\mathrm{single}}$ with corresponding (fixed, but unknown) states $\{\rho_S^x\}_x$ and measurement dilation $( \sigma_M,\{\Pi_{SM}^a\}_a)$. Formally, this implies that the $n_r$-round state and measurements are
\begin{align*}
    {\rho_{S_rM_r}^\mathbf{x}}=\bigotimes_{i=1}^{n_r}\rho^{x_i}_S\otimes \sigma_M\quad\text{and}\quad \Pi_{S_rM_r}^\mathbf{a}=\bigotimes_{i=1}^{n_r} \Pi_{SM}^a.
\end{align*}
Under this IID assumption, Thm.~\ref{thm: single-bit extraction security proof} yields
\begin{align*}
    &\left\lVert \rho_{\mathbf{K}E} - \rho_{\mathbf{K}E}^{\text{ideal}} \right\rVert_1 \leq \Tr\Bigg[\sigma_{M_r}\prod_{i=1}^{n_r}G_{M_i}\Bigg]&&\nonumber \\
    &\qquad= \prod_{i=1}^{n_r}\Bigg(2\sum_{a,x} \nu_{a,x} \Tr\Big[(\rho_{S}^x\otimes \sigma_{M})\Pi_{SM}^a\Big] - 1\Bigg)&&\nonumber\\
   &\qquad= \left(2\sum_{a,x} \nu_{a,x} p(a|x) - 1\right)^{n_r}. &&\label{eq: single-bit special case}
\end{align*}
with $p(a|x)=\Tr[\Pi_{SM}^a(\rho_S^x\otimes \sigma_M)].$ Notice that as long as the dual objective is strictly less than $1$ (i.e., that the outcomes for input $x$ are not completely predictable), the trace distance goes to zero as $n_r\rightarrow\infty$. Equivalent results can be found within Setts.~\ref{set:untrusted_prep} and~\ref{set:untrusted_meas}.

\paragraph{Multi-bit extraction.} Single-bit protocols fall short for most applications requiring useful, high-quality randomness. The following theorem proves that the family of functions $g:\{0,1\}^{n_r}\to\{0,1\}^m$ considered in \cite{Foreman_2025} are extractors in our setting.

\begin{theorem}\label{thm: multi-bit extraction security proof}
    Let $g:\{0,1\}^{n_r} \rightarrow \{0,1\}^m$ be such that:  
    \begin{equation}\label{eq:property extractors}
    \Bigg|\sum_{\mathbf{a}\in\{0,1\}^{n_r}}\Big(\delta_{\mathbf{k},g(\mathbf{a})}-2^{-m}\Big)(-1)^{\mathbf{a}\cdot\mathbf{r}}\Bigg| \leq n_r^2 \: \sqrt{2}^{n_r - m}.
    \end{equation}
    Let $\rho_{\mathbf{A}E}$ be the cq-state that results from $n_r$ sequential uses of a PM device satisfying assumptions \ref{ass:uncorr-prep-and-meas}-\ref{ass:indep-prep} as described by Eq.~\eqref{eq:cq-state}, for a given overlap bound $\delta\in[0,1]$ in the case of Sett.~\ref{set:SDI}. Then, $g$ is a $(m,\epsilon_n)$-deterministic randomness extractor for $\rho_{\mathbf{A}E}$ with 
    \begin{equation*}
        \epsilon_n = n_r^2 \: \sqrt{2}^{m-n_r-2}\Tr\Bigg[\phi\prod_{i=1}^{n_r}\Big(\mathbbm{1}_{i} + G_{i}\Big)\Bigg].
    \end{equation*}
    where $\phi$ and $G_i$ are the same as in Thm.~\ref{thm: single-bit extraction security proof}.
\end{theorem}
The proof of Thm.~\ref{thm: multi-bit extraction security proof} is provided in App.~\ref{app: proof of multi-bit extraction security proof}. In \cite[Lemma 3]{Foreman_2025}, functions satisfying Eq. \eqref{eq:property extractors} are shown to exist whenever $n_r>5$ and $n_r>m$. The proof, however, uses a probabilistic argument and so no explicit constructions are obtained. If we now examine the IID scenario within Sett.~\ref{set:SDI}, in this case Thm.~\ref{thm: multi-bit extraction security proof} yields
\begin{equation*}
    \left\lVert \rho_{\mathbf{K}E} - \rho_{\mathbf{K}E}^{\text{ideal}}\right\rVert_1 \leq  n_r^2 \sqrt{2}^m \left(\frac{2\sum_{a,x} \nu_{a,x} p(a|x)}{\sqrt{2}}\right)^{n_r}.
\end{equation*}
Note that the convergence to zero with $n_r$ now only occurs when the dual SDP objective is within the interval $[1/2,\sqrt{2}/2)$. We do not consider $\sqrt{2}/2$ to be a fundamental bound, but rather an artifact of the proof technique. Equivalent results can again be found for Setts.~\ref{set:untrusted_prep} and~\ref{set:untrusted_meas}. 

\section{Private randomness generation protocol}\label{sec:protocol}
We now focus on Sett.~\ref{set:SDI} to derive a randomness generation protocol based on the previously defined extractors. The upper bounds to the trace distance in Thms.~\ref{thm: single-bit extraction security proof} and~\ref{thm: multi-bit extraction security proof} depend on the expectation value of an operator over the multi-round state of the measurement device. Since in Sett.~\ref{set:SDI} this state is unknown, we need to estimate the expectation value through the input-output statistics of the device $\mathcal{D}$. For this aim, we implement a \emph{spot-checking} randomness generation protocol, where a random subset of rounds is allocated to estimating the relevant expectation values, while the remaining rounds are used to generate the raw key. In this section, we adapt the spot-checking procedure given in~\cite{Foreman_2025} to our setting.

Under assumptions~\ref{ass:uncorr-prep-and-meas}-\ref{ass:indep-prep}, assuming that we have a trusted classical computer to carry out the statistical analysis and considering that all devices (including the adversary's) operate under the laws of quantum physics, the protocol structure is the following:

\begin{algorithm}[H]
	\floatname{algorithm}{Protocol}
	\raggedright
	\caption{Randomness Generation Protocol}
	\label{protocol:rap}
	\begin{algorithmic}[1]
		\STATEx \textbf{Arguments:} 
		\STATEx\hspace{\algorithmicindent} $\delta\in[0,1]$ -- an overlap lower-bound.
		\STATEx\hspace{\algorithmicindent} $D_\delta$ -- a PM device within Sett.~\ref{set:SDI} satisfying assumptions ~\ref{ass:uncorr-prep-and-meas}-\ref{ass:input-randomness} for the given $\delta$. 
		\STATEx\hspace{\algorithmicindent} $n \in \mathbb{N}_+$ -- number of rounds.
		\STATEx\hspace{\algorithmicindent} $p_e\in(0,1)$ -- the probability for an estimation round.
        \STATEx\hspace{\algorithmicindent} $\epsilon\in(0,1]$ -- the protocol's security parameter.
        \STATEx\hspace{\algorithmicindent} $\mathrm{ExtType}\in{\footnotesize\{\texttt{single-bit},\texttt{multi-bit}\}}$ -- the type of extractor to apply.
	
		\STATEx
	
		\STATEx \textbf{Testing and raw data generation:}
		\STATE For every round $i \in \{1,\dots,n\}$ do:
		\STATE\hspace{\algorithmicindent} Choose $t_i \in {\footnotesize\{\texttt{estimation},\texttt{raw key}\}}$ such that $\Pr[T_i={\footnotesize\texttt{estimation}}]=p_e$.
		\STATE\hspace{\algorithmicindent} If $t_i = \texttt{estimation}$, choose $x_i\in\{0,1\}$ uniformly at random and feed it to $\mathcal{D}_\delta$. Record the input-output tuple $z_i=(a_i, x_i)$.

        \STATE\hspace{\algorithmicindent} Else-If $t_i = \texttt{raw key}$, feed $x_i=0$ to $\mathcal{D}_\delta$ and record the output $a_i$.

        \STATEx

        \STATEx \textbf{Randomness extraction:} 
        
        \STATE\hspace{\algorithmicindent} Let $\mathbf{r}=(a_i)_{i\in\{i|t_i=\texttt{raw key}\}}$ be the raw outcomes.
        \STATE\hspace{\algorithmicindent} If $\mathrm{ExtType} = \texttt{single-bit}$, set the protocol's output $k=\texttt{XOR}(\mathbf{r})$ if $m_\mathrm{xor}(\mathbf{t},\mathbf{z})$=1 else output nothing.

        \STATE\hspace{\algorithmicindent} If $\mathrm{ExtType} = \texttt{multi-bit}$, set the protocol's output $\mathbf{k}=g(\mathbf{r})$ for $g:\{0,1\}^{n_r}\to\{0,1\}^m$ a function satisfying Eq. \eqref{eq:property extractors} if $m=m_\mathrm{mul}(\mathbf{t},\mathbf{z})>0$ else output nothing. 
	\end{algorithmic}
\end{algorithm}

The functions $m_\mathrm{xor}$ and $m_\mathrm{mul}$ are defined as follows:

\paragraph{\textrm{XOR} extraction.}
\begin{equation}\label{eq:length single-bit extractor}
m_\mathrm{xor}(\mathbf{t},\mathbf{z}) = \begin{cases}\begin{split}
1 \quad \text{if} &\operatorname*{max} \limits_{\substack{\{\alpha_z\},\ \beta \\ \{\nu_z\},\ \{H_\lambda\}
    }} f_{\mathrm{xor}}(\mathbf{t},\mathbf{z},\{\alpha_{z}\}_{z},\beta,\epsilon) \\ &\geq 0, \end{split}\\
0 \quad \text{otherwise},
\end{cases}
\end{equation}
for 
\begin{align*}
  f_{\mathrm{xor}}(\mathbf{t},\mathbf{z},\{\alpha_{z}\}_{z},\beta,\epsilon) = & \sum_{j=1}^{n_e} \alpha_{z_j} + (\beta - 1)n_r \nonumber\\ & \qquad\qquad\qquad - 2\log(1/\epsilon),
\end{align*}
with $\langle \{\nu_z\},\ \{H_\lambda\}\rangle$ subject to the constraints in Eq. \eqref{eq:dual} and 
\begin{equation}\label{eq:constraints single-bit extractor length}
    \forall z\in\{0,1\}^2,\quad p_r\sqrt{2}^{\beta - 1} (4\nu_{z} - 1) + p_e\sqrt{2}^{\alpha_{z}} = 1,
\end{equation}
where $p_r = 1 - p_e$.\\

\paragraph{Multi-bit extraction.}
\begin{equation}\label{eq:length m-bit extractors}
    m_\mathrm{mul}(\mathbf{t},\mathbf{z}) = \Bigg\lfloor \operatorname*{max} \limits_{\substack{\{\alpha_z\},\ \beta \\ \{\nu_z\},\ \{H_\lambda\}}} f_{\mathrm{mul}}(\mathbf{t},\mathbf{z},\{\alpha_{z}\}_{z},\beta,\epsilon) \Bigg\rfloor,
\end{equation}
for
\begin{align*}
    f_{\mathrm{mul}}(\mathbf{t},\mathbf{z},\{\alpha_{z}\}_{z},\beta,\epsilon)  = & \sum_{j=1}^{n_e} \alpha_{z_j} + \beta n_r \\ & - 2\log(1/\epsilon) - 4\log(n_r),
\end{align*}
with $\langle \{\nu_z\},\ \{H_\lambda\}\rangle$ subject to the constraints in Eq. \eqref{eq:dual} and 
\begin{equation}\label{eq:constraints m-bit extractor length}
    \forall z\in\{0,1\}^2,\quad 4p_r\sqrt{2}^{\beta - 1}\nu_{z} + p_e\sqrt{2}^{\alpha_{z}} = 1.
\end{equation}

Notice that Protocol \ref{protocol:rap} does not include an explicit aborting condition. Instead, when $m(\mathbf{t},\mathbf{z}) = 0$ no output is produced, reflecting estimation data that fails to satisfy the necessary security criteria. Also note that the variables $\langle \{\nu_z\},\ \{H_\lambda\}\rangle$ do not appear explicitly in the objective functions $f_{\mathrm{xor}}$ and $f_{\mathrm{mul}}$ since they are only used to check for feasibility. \\

The security of the protocol is stated in the following theorem, the proof of which is provided in App.~\ref{app:theorem spot-checking}.

\begin{theorem}\label{thm:spot-checking protocol security}
    For any choice of $n$, $p_e$ and $\epsilon$, the above spot-checking protocol generates a final output state $\rho_{\mathbf{K}E|\mathbf{t},\mathbf{z}}$ satisfying the following security condition:
    \begin{equation}
        \frac{1}{2}\sum_{\mathbf{t},\mathbf{z}} p(\mathbf{t},\mathbf{z}) \left\lVert \rho_{\mathbf{K}E|\mathbf{t},\mathbf{z}} - \rho_{\mathbf{K}E|\mathbf{t},\mathbf{z}}^{\text{ideal}} \right\rVert_1 \leq \epsilon,
    \end{equation}
    where $|\mathbf{K}|=m(\mathbf{t},\mathbf{z})$ for given $\mathbf{t}$ and $\mathbf{z}$ is defined in Eq.~\eqref{eq:length single-bit extractor} for the \texttt{XOR} extractor, or in Eq.~\eqref{eq:length m-bit extractors} for the $m$-bit extractors.
\end{theorem}

\section{Numerical simulations}\label{sec:numerics}
Randomness generation protocols in the overlap semi-DI framework presented in Sett.~\ref{set:SDI} have been studied based on the statistics of: unambiguous state discrimination, minimum error state discrimination, and maximum confidence discrimination experiments~\cite{Brask_2017,roch2022quantum,carceller2025improving}. In this section, we evaluate our protocol's performance on a family of behaviors which, to our knowledge, has not been considered in the literature. Namely, for a given overlap assumption $\delta$, the behaviors $\mathbf{p}_\delta$ we consider are defined as
\begin{align}
    p_\delta(a|x):=\begin{cases}
        a(1-2\delta)+\delta & \text{if } x=0,\\
        1-a & \text{if }x=1.
    \end{cases}\label{eq:behaviors}
\end{align}
As we show in a companion publication~\cite{mir2025measurement}, this family of behaviors can be experimentally realized with an optical platform employing polarization-switching phenomena in vertical-cavity surface-emitting lasers (VCSELs). For the purpose of the present work, however, these behaviors are treated as theoretical objects.  

In Fig.~\ref{fig:dual_vs_delta}, we plot the guessing probability in Eq.~\eqref{eq:classical-pguess-povm} for the behaviors in Eq.~\eqref{eq:behaviors}. As can be seen in the plot, when $\delta=1/2$ the unpredicatiblity is maximal, coinciding with a uniform distribution of the outcomes. Whereas, when $\delta\in\{0,1\}$ the behavior becomes completely deterministic and, hence, completely predictable. 

\begin{figure}
    \centering
    \includegraphics[width=\columnwidth]{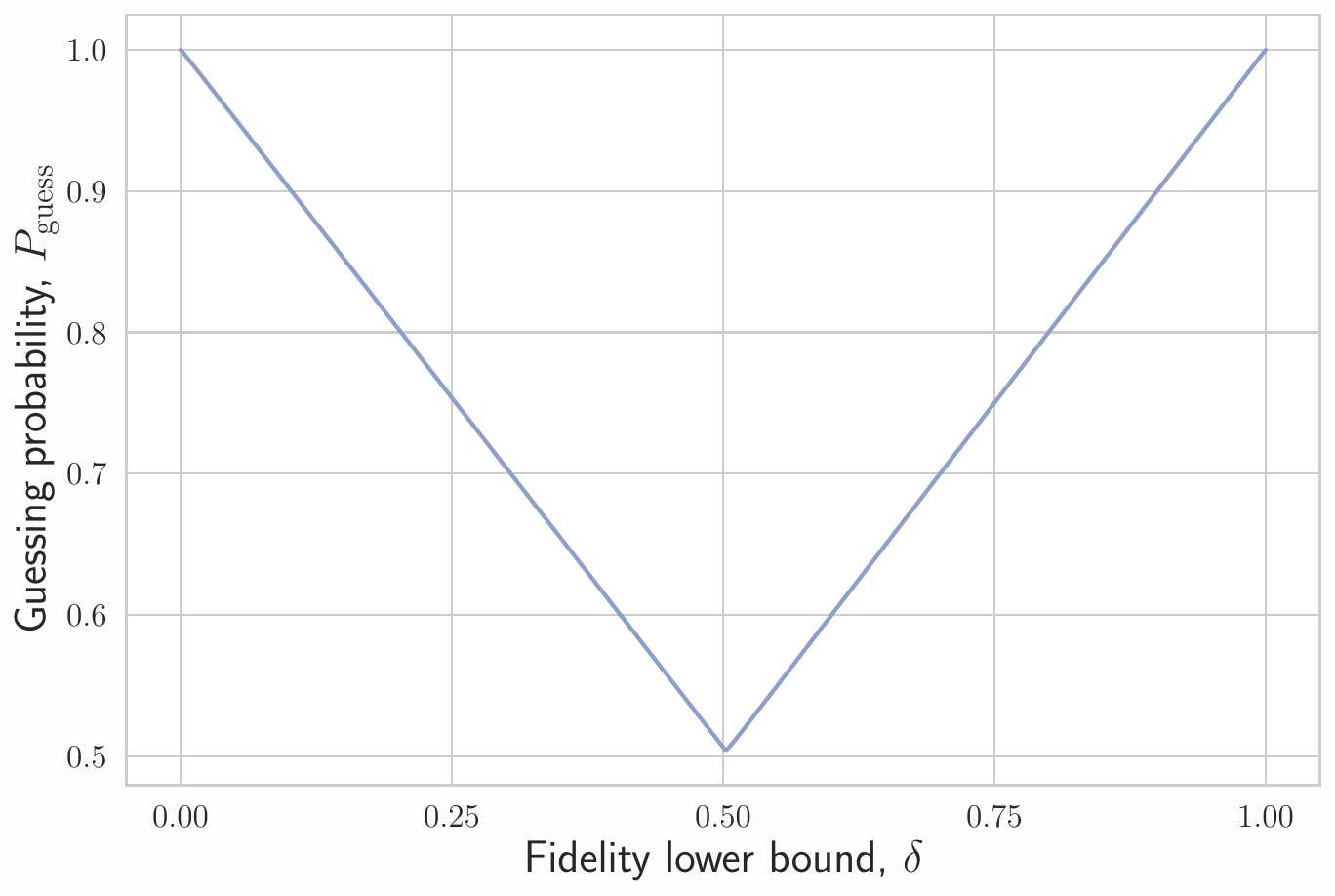}
    \caption{\textbf{Single-round $\boldsymbol{P_{{\rm guess}}}$.} We plot $P_{{\rm guess}}(A|E,X=0,\mathbf{p}_\delta,\delta)$ from Eq.~\eqref{eq:classical-pguess-povm} for the behaviors $\mathbf{p}_\delta$ in Eq.~\eqref{eq:behaviors}.}
    \label{fig:dual_vs_delta}
\end{figure}

In what follows, we study the randomness generation rate of our protocol when executed with honest devices $\mathcal{D}$ whose asymptotic behavior is described by Eq.~\eqref{eq:behaviors}, potentially subject to uncoloured noise (see below). 

Let us first consider the case of multi-bit extraction. Since the protocol's output is of variable length, we study the \emph{expected rate}, that is 
\begin{align}
    \overline{\mathrm{rate}}_\mathrm{mul}(n)=\mathbb{E}[\frac{m_\mathrm{mul}(\mathbf{T},\mathbf{Z})}{n}].\label{eq:finite-size-rate}
\end{align}
We will also consider the \emph{expected asymptotic rate}, i.e., the rate in the limit of infinitely many rounds, given by
\begin{align}
    &\overline{\mathrm{rate}}^\infty_\mathrm{mul}=\lim_{n\to\infty}\overline{\mathrm{rate}}_\mathrm{mul}(n)\nonumber\\
    &\qquad=\max_{\substack{
        \{\alpha_{a,x}\},\ \beta \\
        \{\nu_{a,x}\},\ \{H_\lambda\}
    }}p_e\left(\sum_{a,x} \frac{1}{2}p_\delta(a|x)\alpha_{a,x}\right) + p_r\beta,\label{eq:asym-rate-mul}
\end{align}
with the second equality following from the expression for $m_\mathrm{mul}$ in Eq.~\eqref{eq:length m-bit extractors} and the maximization subject to the constraints in Eq.~\eqref{eq:constraints m-bit extractor length}. Computing these rates involves a nonlinear optimization for which we do not have a general method guaranteeing convergence. Therefore, what we provide in this section are lower bounds to these quantities, resulting from feasible solutions to the optimization problems.

In Fig.~\ref{fig:delta_comparison}, we plot, for different values of the overlap $\delta$: 1) a lower bound to the asymptotic rate in Eq.~\eqref{eq:asym-rate-mul} optimized over $p_e$; and 2) for the best $p_e$ found in 1), lower bounds to the finite rates for different values of $n$. To compute the latter, we simulated the use of an honest IID device with single-round behavior given by Eq.~\eqref{eq:behaviors}. For each value of $n$ up to $10^6$, we calculated the expected finite rate over 100 samples. That is, we sampled from
\begin{align*}
    &p(\mathbf{t},\mathbf{z}=(a_1,x_1)\dots(a_n,x_n))\\&\qquad\qquad\qquad\qquad\qquad=\prod_{i=1}^n \frac{\Pr[T_i=t_i] p_\delta(a_i|x_i)}{2^n}
\end{align*}
100 times and ran the corresponding optimization $m_\mathrm{mul}(\mathbf{t},\mathbf{z})$. Then, we calculated the corresponding average and standard deviation. Since the deviation in the estimation only grows smaller with $n$, for $n$ larger than $10^6$, for which the calculated deviations were always lower than 1\%, we reduced the number of samples to 10. 

\begin{figure}
    \centering
    \includegraphics[width=\columnwidth]{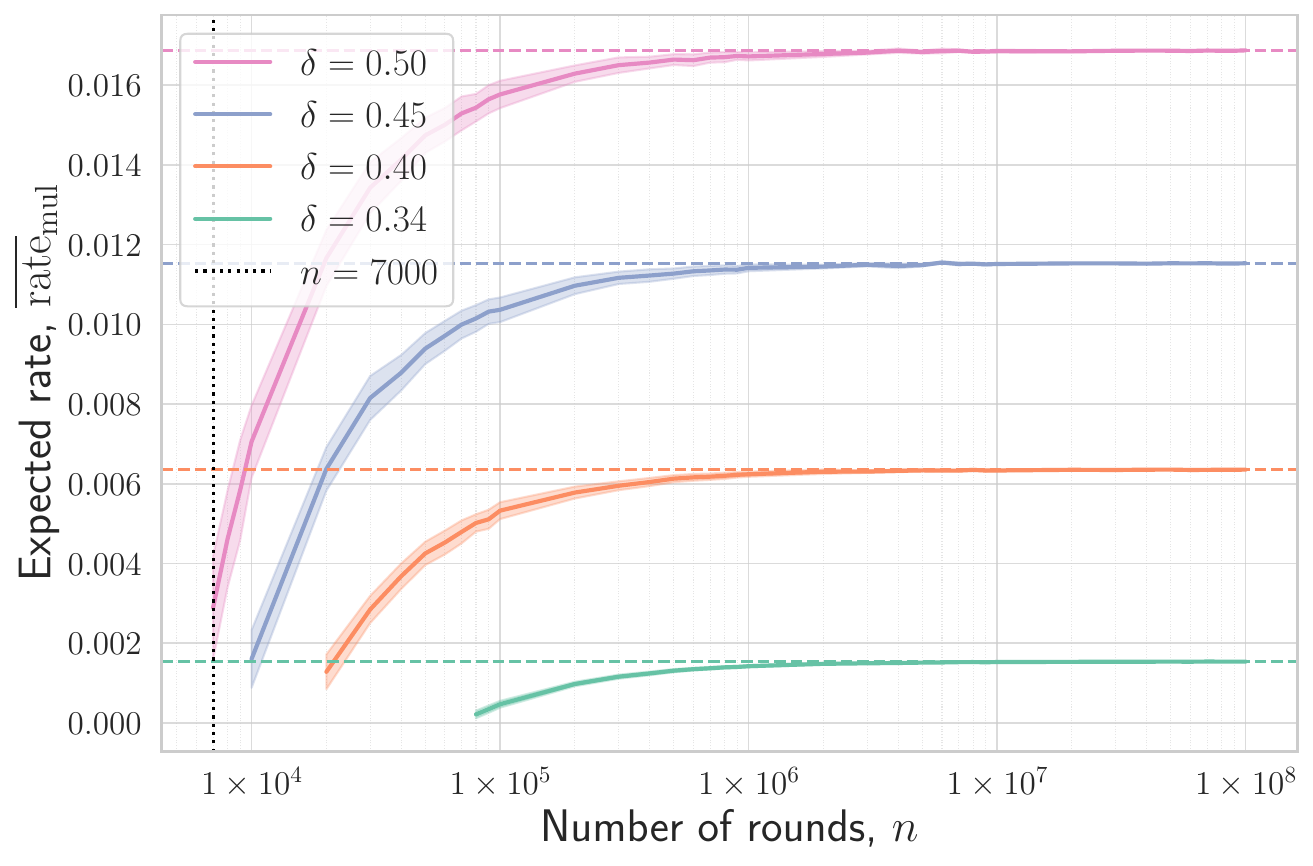}
    \caption{\textbf{Finite-size rates vs. $\boldsymbol{\delta}$}. For different values of the overlap $\delta$, we numerically lower bound the finite-size rates in Eq.~\eqref{eq:finite-size-rate} for an honest IID device with single-round behavior $\mathbf{p}_\delta$ in Eq.~\eqref{eq:behaviors}. Lower bounds to the asymptotic rates given by Eq.~\eqref{eq:asym-rate-mul} are also displayed with dashed lines.}
    \label{fig:delta_comparison}
\end{figure}

We observe that for the value $\delta=1/2$ which optimizes the single-round unpredictability, a positive rate is achievable already for $7\times10^3$ rounds. Additionally, as one might expect due to the strengthening of the assumptions, the best asymptotic rate is an order of magnitude larger than the one reported for the DI protocol from~\cite{Foreman_2025}. On the other hand, under these assumptions, optimal seeded extraction would provide an asymptotic rate larger than or equal to the single-round min-entropy $H_\mathrm{min}=-\log_2 P_\mathrm{guess}$, therefore more than an order of magnitude above what we observe for our seedless protocol (cf. Fig.~\ref{fig:dual_vs_delta}).

Next, we study how the rate is affected by uniform noise. That is, we consider that the protocol is run with an honest IID device with an $n$-round behavior $\mathbf{p}_\delta^\mathrm{noisy}$ such that 
\begin{align}
    p_\delta^\mathrm{noisy}(\mathbf{a}|\mathbf{x})=\prod_{i=1}^{n} \left((1-\gamma)p_\delta(a_i|x_i)+\frac{\gamma}{2}\right)\label{eq:noisy-behavior}
\end{align}

In Fig.~\ref{fig:noisy_comparison}, we plot the finite-size rate for different values of the noise rate $\gamma$ affecting the ideal $\delta=1/2$ behavior. Although, as expected, the rates decrease as the noise increases, the speed at which they do so is not as steep as what has been reported for other quantum cryptographic primitives in the literature. For example, in an implementation of the B92 quantum key distribution protocol~\cite{bennett1992quantum}, the rates reported in~\cite[Fig. 1]{metger2023security} decrease in two orders of magnitude for the same range of the noise parameter as in Fig.~\ref{fig:noisy_comparison}.

\begin{figure}
    \centering
    \includegraphics[width=\columnwidth]{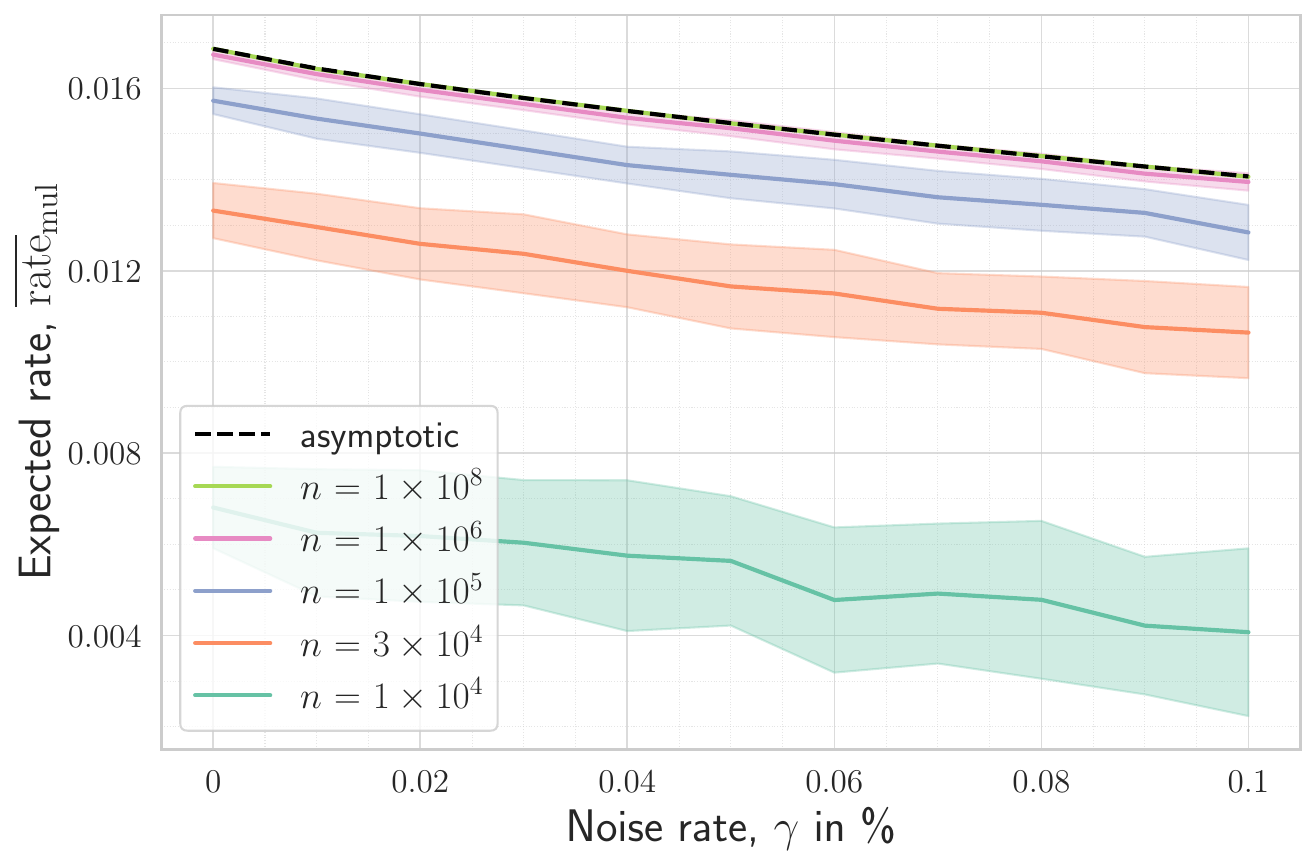}
    \caption{\textbf{Rate's noise tolerance.} For $\delta=1/2$, we plot lower bounds to the finite-size rates achievable by a honest but noisy device with input-output statistics given by Eq.~\eqref{eq:noisy-behavior} and as a function of the uncoloured noise rate $\gamma$.}
    \label{fig:noisy_comparison}
\end{figure}

Regarding the XOR extractor, for the DI construction~\cite{Foreman_2025} the authors reported that, in the limit of infinitely many rounds, in order to have a successful extraction when the protocol's devices are arbitrarily close to the CHSH local bound, the estimation probability had to be $p_e>0.5$. We studied the corresponding problem in our setting. Namely, given an overlap bound $\delta$, we look for the minimum $p_e$ such that 
\begin{align*}\label{eq:large n XOR}
&\lim_{n\to\infty}\mathbb{E}[f_\mathrm{xor}(\mathbf{T},\mathbf{Z})/n]&&\nonumber\\
& =\max_{\substack{
        \{\alpha_{a,x}\},\ \beta \\
        \{\nu_{a,x}\},\ \{H_\lambda\}
    }} p_e \Big(\sum_{a,x} \frac{1}{2}p_\delta(a|x)\alpha_{a,x}\Big) + p_r(\beta - 1)&&\nonumber\\
    &\quad\geq 0.&&
\end{align*}
In contrast with the DI case, we numerically observe that a single bit arbitrarily close to ideal can be extracted for every $p_e>0$ as long as $\delta\not\in\{0,1\}$, i.e, as long as $\mathbf{p}_\delta$ is not completely predictable (c.f. Fig.~\ref{fig:dual_vs_delta}).

\section{Conclusions}\label{sec:conclusions}

We have shown that the deterministic extractors studied in~\cite{Foreman_2025} for the DI setting can be applied in multiple prepare-and-measure frameworks with asymmetric partial trust, namely: a trusted-measurement setting, which may be device-dependent; a device-dependent, trusted-preparation setting; and the semi-DI framework with bounded state overlap developed in~\cite{Brask_2017}. Our main technical result connects the universally composable notion of security based on the trace distance to the expectation value of an operator constructed from dual solutions to the SDP for Eve's guessing probability. We expect that, in future works, this construction may be extended to other semi-DI frameworks and possibly to a general device-dependent setting with zero trust. As for the resulting randomness generation protocol, which we derive for the semi-DI framework, the small number of rounds required for a positive key rate and the noise tolerance observed are promising. On the other hand, the rates achieved are significantly below what seeded extractors can provide for the same setting. Therefore, a significant improvement in the protocol's efficiency is needed before it can become practically relevant for an industrial application.

\section*{Acknowledgements}

We thank Cameron Foreman for his insightful comments and suggestions. TFM and GS acknowledge funding from the Government of Spain (Severo Ochoa CEX2019-000910-S, FUNQIP, NextGeneration EU PRTR-C17.I1 and Torres Quevedo PTQ2021-011870) and European Union (QSNP 101114043 and Quantera Veriqtas), Fundació Cellex, Fundació Mir-Puig and Generalitat de Catalunya (CERCA program). TFM also acknowledges support from Pla de Doctorats Industrials del Departament de Recerca i Universitats de la Generalitat de Catalunya.

\section*{Data and code availability}
The Python code used to perform the numerical optimizations and generate the figures presented in this work is available at \url{https://github.com/QuanTomOptics/SDISeedlessRandExt}.

\typeout{}
\bibliographystyle{quantum}
\bibliography{References}
\clearpage
\onecolumngrid
\appendix

\section{Dual of the SDP for Eve's guessing probability in Eq. \ref{eq:classical-pguess-povm}} \label{app: proof of bound on p_g}
Let us restate the primal from Eq.~\eqref{eq:classical-pguess-povm} in the following equivalent form
\begin{equation}\label{eq:primal_app}
    P_{\mathrm{guess}}(A|E, X=0, \mathbf{p}, \delta) = \max_{\{N_{a|\lambda}\}_{a,\lambda}} \sum_{a=0}^1 \bra{\psi^0}N_{a|a}\ket{\psi^0},
\end{equation}
subject to
\begin{align}
    &\forall x,a,\quad\sum_{\lambda=0}^1 \bra{\psi^x}N_{a|\lambda} \ket{\psi^x} = p(a|x),\\
    &N_{a|\lambda} \succeq 0,\\
    &\sum_{a=0}^1 N_{a|\lambda} = p(\lambda) \mathbbm{1} = \frac{1}{2}\Tr\Big[\sum_{a=0}^1 N_{a|\lambda}\Big] \mathbbm{1},   
\end{align}
with $\ket{\psi^0}=\ket{0}$ and $\ket{\psi^1}=\sqrt{\delta}\ket{0}+\sqrt{1-\delta}\ket{1}$, and where we have absorbed the probabilities $p(\lambda)$ into the now unnormalized POVMs $\{N_{a|\lambda}\}$. To obtain the dual SDP, a Lagrangian is defined, which is a function of the primal SDP's variables $N_{a|\lambda}$ as well as Lagrange multipliers variables for each of the constraints in the primal SDP; namely, scalars $\nu_{a,x}$, and Hermitian $2\times2$ matrices $H_\lambda$ and $G_{a|\lambda}$, such that
\begin{align}
    \mathcal{L} &= \sum_{\lambda=0}^1 \delta_{a,\lambda}\bra{\psi^0} N_{a|\lambda} \ket{\psi^0} + \sum_{a,x=0}^1 \nu_{a,x} \left(p(a|x) - \sum_{\lambda=0}^1 \bra{\psi^x}N_{a|\lambda}\ket{\psi^x}\right) \nonumber \\
    &+ \sum_{\lambda=0}^{1} \Tr\Bigg[H_\lambda\left(\sum_{a=0}^1 N_{a|\lambda} - \frac{1}{2}\Tr\Big[\sum_{a=0}^1 N_{a|\lambda}\Big]\mathbbm{1}\right)\Bigg] + \sum_{a,\lambda=0}^{1} \Tr\Big[G_{a|\lambda}N_{a|\lambda}\Big].
\end{align}

The supremum $\mathcal{S}$ of the Lagrangian over the primal SDP variables is $\mathcal{S} \coloneqq \text{sup}_{\{N_{a|\lambda}\}_{a,\lambda}
} \mathcal{L}$. Note that, for any particular solution $\{N_{a|\lambda}\}_{a,\lambda}$ of the primal SDP, the two middle terms of the Lagrangian vanish due to the contraints of the primal. On the other hand, the first term of the Lagrangian is the objective of the primal SDP, and the last term is non-negative as long as $G_{a|\lambda} \succeq 0$, due to the second constraint of the primal and the fact that the trace of the product of two positive semidefinite operators is non-negative. Therefore, enforcing $G_{a|\lambda} \succeq 0$ guarantees that $P_{\mathrm{guess}} \leq \mathcal{S}$. To get useful bounds, $\mathcal{S}$ must be minimized over the dual SDP variables. However, $\mathcal{S}$ can be first rewritten in a more convenient form as
\begin{equation}
    \mathcal{S} = \text{sup}_{\{N_{a|\lambda}\}_{a,\lambda}} \sum_{a,\lambda=0}^1 \Tr\Big[N_{a|\lambda} K_{a|\lambda}\Big] + \sum_{a,x=0}^1 \nu_{a,x}p(a|x),
\end{equation}
where 
\begin{equation}
    K_{a|\lambda} \coloneqq \dyad{\psi^0}\delta_{a,\lambda} - \sum_{x=0}^1 \dyad{\psi^x} \nu_{a,x} + G_{a|\lambda} + H_\lambda - \frac{1}{2}\Tr\Big[H_\lambda\Big]\mathbbm{1}.
\end{equation}

Any values for the dual SDP variables can be chosen; in particular, $K_{a|\lambda} = 0$ greatly simplifies the problem. This still results in a suitable upper bound, albeit potentially suboptimal. Given that $G_{a|\lambda} \succeq 0$, setting $K_{a|\lambda} = 0$ implies that the remaining terms in the definition of $K_{a|\lambda}$ must be negative semidefinite. With this, the final form of the dual SDP reads
\begin{equation}
    P_{\mathrm{guess}}^\ast (A|E, X=0,\mathbf{p},\delta)\coloneqq \min_{\{H_\lambda\}_{\lambda},\{\nu_{a,x}\}_{a,x}} \sum_{a,x=0}^1 \nu_{a,x}p(a|x),
\end{equation}
subject to
\begin{equation}
    \quad \dyad{\psi^0}\delta_{a,\lambda} - \sum_{x=0}^1 \dyad{\psi^x} \nu_{a,x} + H_\lambda - \frac{1}{2}\Tr\Big[H_\lambda\Big]\mathbbm{1} \preceq 0,
\end{equation}
which we rewrite in an equivalent form in Eq.~\eqref{eq:dual}. The optimal value of this dual SDP yields an upper bound on Eve's optimal guessing probability, as $P_{\mathrm{guess}} \leq P_{\mathrm{guess}}^\ast$.

\section{Primal and dual SDP formulations for Eve's guessing probability in Setts.~\ref{set:untrusted_prep} and~\ref{set:untrusted_meas}}\label{app: extra sdps}
We start by stating the primal SDP form of Eve's guessing probability for Sett.~\ref{set:untrusted_meas}, due to its similarities with Sett.~\ref{set:SDI}, in a similar manner as in App.~\ref{app: proof of bound on p_g}, which reads
\begin{equation}
    P_{\mathrm{guess}}(A|E, \mathbf{p}) = \max_{\{N_{a|\lambda}\}_{a,\lambda}} \sum_{a=0}^1 \Tr[N_{a|a}\rho],
\end{equation}
subject to
\begin{align}
    &\forall k,a,\quad\sum_{\lambda=0}^1\Tr[N_{a|\lambda}\varphi^k] = \Tr[M_a\varphi^k] = p(a|k),\\
    &N_{a|\lambda} \succeq 0,\\
    &\sum_{a=0}^1 N_{a|\lambda} = p(\lambda) \mathbbm{1} = \frac{1}{2}\Tr\Big[\sum_{a=0}^1 N_{a|\lambda}\Big] \mathbbm{1},   
\end{align}
where $\mathbf{p}\in \mathcal{Q}_{\mathrm{P}}$, $\rho$ is the trusted and characterized preparation, $\{N_{a|\lambda}\}_a$ are some unnormalized POVMs that reconstruct our untrusted but characterized POVM $\{M_a\}_a$, and $\{\varphi^k\}_k$ is a \textit{tomographically complete set of states} that unequivocally defines $\{M_a\}_a$. For simplicity, we restrict to qubit dimensions, where a minimal tomographically complete set has $k\in\{0,1,2,3\}$ elements.

Following a similar derivation to that of App.~\ref{app: proof of bound on p_g}, by defining a Lagrangian $\mathcal{L}$ with some Lagrange multipliers and then taking its supremum over the primal SDP variables, it is straightforward to see that the dual SDP formulation within Sett.~\ref{set:untrusted_meas} is given by
\begin{equation}\label{eq:dual_sett2}
    P_{\mathrm{guess}}^\ast (A|E,\mathbf{p}) = \min_{\{H_\lambda\}_{\lambda},\{\nu_{a,k}\}_{a,k}} \sum_{a=0}^1\sum_{k=0}^3 \nu_{a,k}p(a|k),
\end{equation}
subject to
\begin{equation}
    \rho\delta_{a,\lambda} - \sum_{k=0}^3 \varphi^k \nu_{a,k} + H_\lambda - \frac{1}{2}\Tr\Big[H_\lambda\Big]\mathbbm{1} \preceq 0.
\end{equation}

In the case of Sett.~\ref{set:untrusted_prep}, the situation is slightly different, and the primal SDP now has the form
\begin{equation}
    P_{\mathrm{guess}}(A|E, \mathbf{p}) = \max_{\{\sigma_\lambda\}_{\lambda}} \sum_{a=0}^1 \Tr[M_a\sigma_a],
\end{equation}
subject to
\begin{align}
    &\forall b,\quad\sum_{\lambda=0}^1 \Tr[O_b\sigma_\lambda] = \Tr[O_b\rho] = p(b),\\
    &\sigma_\lambda \succeq 0, 
\end{align}
where $\mathbf{p}\in \mathcal{Q}_{\mathrm{M}}$, $\{M_a\}_a$ is the trusted and characterized measurement, $\{\sigma_\lambda\}_\lambda$ are some unnormalized states that reconstruct our untrusted but characterized state $\rho$, and $\{O_b\}_b$ is an \textit{informationally complete} (IC) POVM that unequivocally defines $\rho$. Restricting to qubit dimensions again for simplicity, we have that $b\in\{0,1,2,3\}$ for a minimal IC-POVM. Note that the normalization constraint for $\{\sigma_\lambda\}_\lambda$ is implied by the first constraint and the fact that $\{O_b\}_b$ sums up to the identity. In the case where the preparation is left uncharacterized, the IC-POVM condition is replaced by $\sum_\lambda\Tr[M_a\sigma_\lambda ]=p(a)$, $\forall a$, and the following derivations are carried out equivalently.

Proceeding as before, the formulation of the Lagrangian in this case is greatly simplified, where we have that
\begin{equation}
    \mathcal{L} = \sum_{\lambda=0}^1\delta_{a,\lambda}\Tr[M_a\sigma_\lambda ] + \sum_{b=0}^3\nu_b\left(p(b) - \sum_{\lambda=0}^1\Tr[O_b\sigma_\lambda] \right) + \sum_{\lambda=0}^1\Tr[G_\lambda\sigma_\lambda],
\end{equation}
which, after taking the supremum over $\{\sigma_\lambda\}_\lambda$, leads to the dual SDP given by
\begin{equation}\label{eq:dual_sett1}
    P_{\mathrm{guess}}^\ast (A|E,\mathbf{p}) = \min_{\{\nu_{b}\}_{b}} \sum_{b=0}^3\nu_{b}p(b),
\end{equation}
subject to
\begin{equation}
    \delta_{a,\lambda}M_a - \sum_{b=0}^3\nu_bO_b \preceq 0.
\end{equation}

\section{Proof of Thm.~\ref{thm: ancilla operator inequality}}\label{app: proof of ancilla operator inequality}
We restate Thm.~\ref{thm: ancilla operator inequality} for ease of read and then prove it.
\setcounter{theorem}{0}
\begin{theorem}
    For every $\delta\in[0,1]$, if $\langle \{\nu_{a,x}\}_{a,x=0}^1,\{H_\lambda\}_{\lambda=0}^1\rangle$ is a solution to Eq. \eqref{eq:dual} for some behavior $\mathbf{p}$, then, for any $\{\Pi_{SM}^a\}_{a=0}^1$ and $\{\rho_{S}^x\}_{x=0}^1$ such that $F(\rho_S^0,\rho_S^1)\geq\delta$, the following inequality over operators in $\mathcal{H}_M$ holds:
    \begin{equation}
        \pm \Tr_S\Big[(\rho_{S}^0 \otimes \mathbbm{1}_{M})(\Pi_{SM}^0-\Pi_{SM}^1)\Big] \preceq G_{M}, 
    \end{equation}
    with $G_M \coloneqq 2\sum_{a,x} \nu_{a,x} \Tr_{S}\Big[(\rho_{S}^x \otimes \mathbbm{1}_{M}) \Pi_{SM}^a\Big] - \mathbbm{1}_{M}$.
\end{theorem}
\begin{proof}
Notice that in the constraints of the dual in Eq. \eqref{eq:dual} the behavior $\mathbf{p}$ does not appear. This implies that any feasible solution $\langle \{H_\lambda\},\{\nu_{a,x}\} \rangle $ for a given behavior $\mathbf{p}$ and a given $\delta\in[0,1]$ is also a feasible solution for any other behavior $\mathbf{p'}\in\Qover$. Therefore, Eq. \eqref{eq: predictability bound} generalizes to
\begin{equation}
    \forall\mathbf{p'}\in\Qover,\quad\Big| p'(0|0) - p'(1|0) \Big| \leq 2\sum_{a,x=0}^1 \nu_{a,x} p'(a|x) - 1.
\end{equation}
In particular, for $p'(a|x)=\Tr\left[(\rho_{S}^x \otimes \sigma_{M})\Pi_{SM}^a\right]$, we have
\begin{equation}
    \pm \Tr\Big[(\rho_{S}^0 \otimes \sigma_{M})(\Pi_{SM}^0 - \Pi_{SM}^1)\Big] \leq 2\sum_{a,x=0}^1 \nu_{a,x} \Tr\Big[(\rho_S^x \otimes \sigma_{M})\Pi_{SM}^a\Big] - 1.
\end{equation}
Splitting the total trace into partial traces, we have that
\begin{align}
    &\pm \Tr_{M}\Bigg[\sigma_{M} \Tr_{S}\Big[(\rho_{S}^0 \otimes \mathbbm{1}_{M})(\Pi_{SM}^0 - \Pi_{SM}^1)\Big]\Bigg] \leq&& \nonumber \\
    &\qquad\qquad\qquad\qquad\qquad\qquad\Tr_{M}\Bigg[\sigma_{M}\left(2\sum_{a,x=0}^1 \nu_{a,x} \Tr_{S}\Big[(\rho_S^x \otimes \mathbbm{1}_{M})\Pi_{SM}^a\Big] - \mathbbm{1}_M\right)\Bigg].&&
\end{align}
Since this inequality holds for an arbitrary $\sigma_{M} \in M$, the operators inside the partial trace over $M$ must satisfy
\begin{equation}
    \pm\Tr_{S}\Big[(\rho_{S}^0 \otimes \mathbbm{1}_{M})(\Pi_{SM}^0 - \Pi_{SM}^1)\Big] \preceq 2\sum_{a,x=0}^1 \nu_{a,x} \Tr_{S}\Big[(\rho_S^x \otimes \mathbbm{1}_{M})\Pi_{SM}^a\Big] - \mathbbm{1}_M \coloneqq G_M.
\end{equation}
\end{proof}

\section{Formulations of Thm.~\ref{thm: ancilla operator inequality} within Setts.~\ref{set:untrusted_prep} and~\ref{set:untrusted_meas}}\label{app: extension of ancilla operator inequality}

By taking a look at the bound in Eq.~\eqref{eq:dual_sett1}, corresponding to Sett.~\ref{set:untrusted_prep}, and the bound in Eq.~\eqref{eq:dual_sett2}, corresponding to Sett.~\ref{set:untrusted_meas}, it is straightforward to see that they have the same structure as the one in Eq~\eqref{eq:dual}. Thus, they also provide bounds on the predictability of the outcomes similarly to Eq.~\eqref{eq: predictability bound}. Specifically, given some feasible solutions to Eqs.~\eqref{eq:dual_sett1} and~\eqref{eq:dual_sett2} for given behaviors $\mathbf{p}_\mathrm{M}\in\mathcal{Q}_\mathrm{M}$ and $\mathbf{p}_\mathrm{P}\in\mathcal{Q}_\mathrm{P}$, respectively, we have in general that
\begin{align}
    \forall\mathbf{p}'_{\mathrm{M}}\in\mathcal{Q}_{\mathrm{M}},&\quad\Big| p'(0) - p'(1) \Big| \leq 2\sum_{a=0}^1\sum_{k=0}^3 \nu_{a,k} p'(a|k) - 1\text{ \ref{set:untrusted_prep}}, \nonumber\\
    \forall\mathbf{p}'_{\mathrm{P}}\in\mathcal{Q}_{\mathrm{P}},&\quad\Big| p'(0) - p'(1) \Big| \leq 2\sum_{b=0}^3 \nu_{b} p'(b) - 1 \text{ \ref{set:untrusted_meas}},
\end{align}
where, in an abuse of notation, we are using the same $p$ to denote the behavior of the device in a randomness generation stage, $p(a) = \Tr[\rho M_a]$, and in a characterization stage, $p(b) = \Tr[O_b\rho]$~\ref{set:untrusted_prep} or $p(a|k) = \Tr[M_a\varphi^k]$~\ref{set:untrusted_meas}. Introducing the Naimark dilation picture, we can write these bounds in operational form as
\begin{align}
    &\pm \Tr\Big[(\rho_{S} \otimes \dyad{0}_M)(\Pi_{SM}^0 - \Pi_{SM}^1)\Big] \leq 2\sum_{b=0}^3 \nu_{b} \Tr\Big[(\rho_S \otimes \dyad{0}_{M})\tilde{\Pi}_{SM}^b\Big] - 1 \text{ \ref{set:untrusted_prep}}, \nonumber\\
    &\pm \Tr\Big[(\rho_{S} \otimes \sigma_{M})(\Pi_{SM}^0 - \Pi_{SM}^1)\Big] \leq 2\sum_{a=0}^1\sum_{k=0}^3 \nu_{a,k} \Tr\Big[(\varphi^k_S \otimes \sigma_{M})\Pi_{SM}^a\Big] - 1 \text{ \ref{set:untrusted_meas}},
\end{align}
where $(\dyad{0}_M,\{\Pi_a\}_a)$ is a standard Naimark dilation of $\{M_a\}_a$, and $(\dyad{0}_M,\{\tilde{\Pi}_a\}_a)$ is a standard Naimark dilation of the IC-POVM $\{O_b\}_b$.

Following the same steps as in App.~\ref{app: proof of ancilla operator inequality} by splitting into partial traces, where in the case of Sett.~\ref{set:untrusted_prep} the outside partial trace is over system $S$ instead of system $M$, we obtain the following reformulations of Thm.~\ref{thm: ancilla operator inequality}.

\setcounter{theorem}{0}
\begin{theorem}[\ref{set:untrusted_prep}]
    If $\langle \{\nu_{b}\}_{b=0}^3\rangle$ is a solution to Eq.~\eqref{eq:dual_sett1} for some behavior $\mathbf{p}\in\mathcal{Q}_\mathrm{M}$, then, for any standard Naimark dilation $(\dyad{0}_M,\{\Pi_{SM}^a\}_{a=0}^1)$ of $\{M_a\}_{a=0}^1$ and any standard Naimark dilation $(\dyad{0}_M,\{\tilde{\Pi}_{SM}^b\}_{b=0}^3)$ of an IC-POVM $\{O_b\}_{b=0}^3$, the following inequality over operators in $\mathcal{H}_S$ holds:
    \begin{equation}
    \pm \Tr_M\Big[(\mathbbm{1}_S \otimes \dyad{0}_M)(\Pi_{SM}^0-\Pi_{SM}^1)\Big] \preceq G_{S}, 
\end{equation}
with $G_S \coloneqq 2\sum_{b} \nu_{b} \Tr_{M}\Big[(\mathbbm{1}_S \otimes \dyad{0}_{M}) \tilde{\Pi}_{SM}^b\Big] - \mathbbm{1}_{S}$.
\end{theorem}

\setcounter{theorem}{0}
\begin{theorem}[\ref{set:untrusted_meas}]
    If $\langle \{\nu_{a,k}\}_{a=0}^1{}_{k=0}^3,\{H_\lambda\}_{\lambda=0}^1\rangle$ is a solution to Eq.~\eqref{eq:dual_sett2} for some behavior $\mathbf{p}\in\mathcal{Q}_\mathrm{P}$, then, for any $\{\Pi_{SM}^a\}_{a=0}^1$, $\rho_S$, and tomographically complete set of states $\{\varphi^k\}_{k=0}^3$, the following inequality over operators in $\mathcal{H}_M$ holds:
    \begin{equation}
    \pm \Tr_S\Big[(\rho_{S} \otimes \mathbbm{1}_M)(\Pi_{SM}^0-\Pi_{SM}^1)\Big] \preceq G_{M}, 
\end{equation}
with $G_M \coloneqq 2\sum_{a,k} \nu_{a,k} \Tr_{S}\Big[(\varphi^k_{S} \otimes \mathbbm{1}_{M}) \Pi_{SM}^a\Big] - \mathbbm{1}_{M}$.
\end{theorem}

\section{Proof of Thm.~\ref{thm: single-bit extraction security proof}}\label{app: proof of single-bit extraction security proof}

This proof is adapted directly from~\cite{Foreman_2025}, with key modifications to accommodate for Setts.~\ref{set:untrusted_prep}-\ref{set:SDI}. The cq-state defined in Eq.~\eqref{eq: actual cq-state} can be written in the single-bit case as
\begin{equation}\label{eq:initial cq-state}
    \rho_{\mathbf{K}E} = \sum_{k=0}^1 \sum_{\mathbf{a}\in\{0,1\}^{n_r}} \dyad{k} \otimes \delta_{k,\texttt{XOR}(\mathbf{a})} \Tr_{S_rM_r}\Bigg[\rho_{S_rM_rE} \prod_{i=1}^{n_r} \Pi_{S_iM_i}^{a_i}\Bigg],
\end{equation}
where, from Assumptions~\ref{ass:uncorr-prep-and-meas}-\ref{ass:indep-prep}, we have that

\begin{itemize}
    \item $S_r = \bigotimes_{i=1}^{n_r} S_i$ is Alice's system across all raw key generation rounds. In Sett.~\ref{set:untrusted_prep}, it is prepared in a state $\rho_{S_r}$ that may exhibit correlations between rounds. In Setts.~\ref{set:untrusted_meas} and~\ref{set:SDI}, it is prepared in the product states $\rho_{S_r}=\bigotimes_{i=1}^{n_r} \rho_{S_i}$ and $\rho_{S_r} = \bigotimes_{i=1}^{n_r} \rho_{S_i}^{0}$, respectively, where in the latter the input is fixed at $x_i = 0$ for all rounds. This system is assumed to be uncorrelated with the measurement device. 
    \item $M_r = \bigotimes_{i=1}^{n_r} M_i$ represents the internal system of the measurement device. In Sett.~\ref{set:untrusted_prep}, its state is a product state $\sigma_{M_r}=\bigotimes_{i=1}^{n_r}\dyad{0}_{M_i}$. In Setts.~\ref{set:untrusted_meas} and~\ref{set:SDI}, its state $\sigma_{M_r}$ is a general state that may exhibit correlations across rounds.
    \item $E$ is Eve's system, and she holds either a purification of $\rho_{S_r}$ in Sett.~\ref{set:untrusted_prep}, or a purification of $\sigma_{M_r}$ in Setts.~\ref{set:untrusted_meas} and~\ref{set:SDI}.
    \item $\Pi_{S_iM_i}^{a_i}$ denotes the dilated measurement operator corresponding to outcome $a_i$ in round $i$. 
\end{itemize}
The initial global state in the protocol, following Eq.~\eqref{eq:cq-state}, is then given by
\begin{align}\label{eq:initial state}
     \rho_{S_rM_rE} &= \dyad{\Psi}_{S_rE} \otimes  \left(\bigotimes_{i=1}^{n_r} \dyad{0}_{M_i}\right)\text{~\ref{set:untrusted_prep}}, \nonumber\\  \rho_{S_rM_rE} &= \left(\bigotimes_{i=1}^{n_r} \rho_{S_i}\right) \otimes \dyad{\Phi}_{M_rE}\text{~\ref{set:untrusted_meas}}, \nonumber\\
    \rho_{S_rM_rE} &= \left(\bigotimes_{i=1}^{n_r} \rho_{S_i}^{0}\right) \otimes \dyad{\Phi}_{M_rE}\text{~\ref{set:SDI}}.
\end{align}
Due to the similarities between Setts.~\ref{set:untrusted_meas} and~\ref{set:SDI}, for the remainder of the proof we will only focus on Setts.~\ref{set:untrusted_prep} and~\ref{set:untrusted_meas}, as the derivations for Sett.~\ref{set:SDI} follow analogously. Substituting Eq.~\eqref{eq:initial state} into Eq.~\eqref{eq:initial cq-state}, using the linearity of the partial trace, and factoring out the operators that act trivially on the subsystem over which the partial trace is taken, it follows that
\begin{align}\label{eq:modified cq-state}
    \rho_{\mathbf{K}E} &= \sum_{k=0}^1 \dyad{k} \otimes \Tr_{S_r}\Bigg[\left(\sum_{\mathbf{a}} \delta_{k,\texttt{XOR}(\mathbf{a})} \prod_{i=1}^{n_r}\Tr_{M_i}\Big[(\mathbbm{1}_{S_i}\otimes \dyad{0}_{M_i})\Pi_{S_iM_i}^{a_i}\Big]\right)\dyad{\Psi}_{S_rE}\Bigg]\text{~\ref{set:untrusted_prep}}, \nonumber\\
    \rho_{\mathbf{K}E} &= \sum_{k=0}^1 \dyad{k} \otimes \Tr_{M_r}\Bigg[\left(\sum_{\mathbf{a}} \delta_{k,\texttt{XOR}(\mathbf{a})} \prod_{i=1}^{n_r}\Tr_{S_i}\Big[(\rho_{S_i} \otimes \mathbbm{1}_{M_i})\Pi_{S_iM_i}^{a_i}\Big]\right)\dyad{\Phi}_{M_rE}\Bigg]\text{~\ref{set:untrusted_meas}}.
\end{align}
Focusing now on the terms inside the parentheses, we first define the operator
\begin{equation}
    C_{S_iM_i} = \Pi_{S_iM_i}^{0} - \Pi_{S_iM_i}^1,
\end{equation}
so that the projectors can be expressed as
\begin{equation}\label{eq:inverse relationship}
    \Pi_{S_iM_i}^{a_i} = \frac{1}{2}\Big(1 + (-1)^{a_i}C_{S_iM_i}\Big),
\end{equation}
noting that $\Pi_{S_iM_i}^0 + \Pi_{S_iM_i}^1 = \mathbbm{1}_{S_iM_i}$. Substituting this expression inside Eq.~\eqref{eq:modified cq-state}, the terms inside the parentheses become
\begin{align}
    &\frac{1}{2^{n_r}} \sum_{\mathbf{a}} \delta_{k,\texttt{XOR}(\mathbf{a})} \prod_{i=1}^{n_r}\left(\mathbbm{1}_{S_i} + (-1)^{a_i}\Tr_{M_i}\Big[(\mathbbm{1}_{S_i} \otimes \dyad{0}_{M_i})C_{S_iM_i}\Big]\right)\text{~\ref{set:untrusted_prep}}, \nonumber\\
    &\frac{1}{2^{n_r}} \sum_{\mathbf{a}} \delta_{k,\texttt{XOR}(\mathbf{a})} \prod_{i=1}^{n_r}\left(\mathbbm{1}_{M_i} + (-1)^{a_i}\Tr_{S_i}\Big[(\rho_{S_i} \otimes \mathbbm{1}_{M_i})C_{S_iM_i}\Big]\right)\text{~\ref{set:untrusted_meas}}.
\end{align}

Next, note that $C_{S_iM_i}$ is full rank, which means that $(C_{S_iM_i})^0 = \mathbbm{1}_{S_iM_i}$. The purpose of the parenthesis is to avoid confusing the exponent with other labels. Moreover, by rewriting the delta function as
\begin{equation}
    \delta_{k,\texttt{XOR}(\mathbf{a})} = \frac{1}{2}\Big[1 + (-1)^{\mathbf{a}\cdot\mathbf{1} + k}\Big],
\end{equation}
the expressions become
\begin{align}
    &\frac{1}{2^{n_r + 1}} \sum_{\mathbf{a}}\Bigg[\sum_{\mathbf{r}} \left((-1)^{\mathbf{a}\cdot\mathbf{r}} \prod_{i=1}^{n_r}\Tr_{M_i}\Big[(\mathbbm{1}_{S_i} \otimes \dyad{0}_{M_i})(C_{S_iM_i})^{r_i}\Big]\right) \nonumber\\
    &\qquad\qquad\quad+ (-1)^{k}\sum_{\mathbf{r}'}\left((-1)^{\mathbf{a} \cdot (\mathbf{r}' + \mathbf{1})}\prod_{i=1}^{n_r}\Tr_{M_i}\Big[(\mathbbm{1}_{S_i} \otimes \dyad{0}_{M_i})(C_{S_iM_i})^{r_i'}\Big]\right)\Bigg]\text{~\ref{set:untrusted_prep}},\nonumber\\
    &\frac{1}{2^{n_r + 1}} \sum_{\mathbf{a}}\Bigg[\sum_{\mathbf{r}} \left((-1)^{\mathbf{a}\cdot\mathbf{r}} \prod_{i=1}^{n_r}\Tr_{S_i}\Big[(\rho_{S_i} \otimes \mathbbm{1}_{M_i})(C_{S_iM_i})^{r_i}\Big]\right) \nonumber\\
    &\qquad\qquad\quad+ (-1)^{k}\sum_{\mathbf{r}'}\left((-1)^{\mathbf{a} \cdot (\mathbf{r}' + \mathbf{1})}\prod_{i=1}^{n_r}\Tr_{S_i}\Big[(\rho_{S_i} \otimes \mathbbm{1}_{M_i})(C_{S_iM_i})^{r_i'}\Big]\right)\Bigg]\text{~\ref{set:untrusted_meas}}.
\end{align}

By a standard symmetry argument, when performing the summation over $\mathbf{a}$, only the terms $\mathbf{r}=\mathbf{0}$ and $\mathbf{r}'=\mathbf{1}$ survive. Therefore, the expressions become
\begin{align}
    &\frac{1}{2}\Bigg(\mathbbm{1}_{S_r} + (-1)^k\prod_{i=1}^{n_r} \Tr_{M_i}\Big[(\mathbbm{1}_{S_i} \otimes \dyad{0}_{M_i})C_{S_iM_i}\Big]\Bigg)\text{~\ref{set:untrusted_prep}},\nonumber\\
    &\frac{1}{2}\Bigg(\mathbbm{1}_{M_r} + (-1)^k\prod_{i=1}^{n_r} \Tr_{S_i}\Big[(\rho_{S_i} \otimes \mathbbm{1}_{M_i})C_{S_iM_i}\Big]\Bigg)\text{~\ref{set:untrusted_meas}}.
\end{align}
Inserting these expressions into Eq.~\eqref{eq:modified cq-state}, it can be rewritten as
\begin{align}
    &\rho_{\mathbf{K}E} = \frac{1}{2}\sum_{k=0}^1 \dyad{k} \otimes \left(\rho_E+(-1)^k \Tr_{S_r}\Big[\dyad{\Psi}_{S_rE}\prod_{i=1}^{n_r}\Tr_{M_i}\big[(\mathbbm{1}_{S_i} \otimes \dyad{0}_{M_i})C_{S_iM_i}\big]\Big]\right)\text{~\ref{set:untrusted_prep}},\nonumber\\
    &\rho_{\mathbf{K}E} = \frac{1}{2}\sum_{k=0}^1 \dyad{k} \otimes \left(\rho_E+(-1)^k \Tr_{M_r}\Big[\dyad{\Phi}_{M_rE}\prod_{i=1}^{n_r}\Tr_{S_i}\big[(\rho_{S_i} \otimes \mathbbm{1}_{M_i})C_{S_iM_i}\big]\Big]\right)\text{~\ref{set:untrusted_meas}}.
\end{align}
Note that the first term in expressions above gives precisely $\rho_{\mathbf{K}E}^{\text{ideal}} = u_{\mathbf{K}} \otimes \rho_{E}$, from which it is clear that
\begin{align}
    \left\lVert \rho_{\mathbf{K}E} - \rho_{\mathbf{K}E}^{\text{ideal}} \right\rVert_1 &= \frac{1}{2}\left\lVert \sum_{k=0}^1 (-1)^k \dyad{k} \otimes \Tr_{S_r}\Big[\dyad{\Psi}_{S_rE}\prod_{i=1}^{n_r}\Tr_{M_i}\big[(\mathbbm{1}_{S_i} \otimes \dyad{0}_{M_i})C_{S_iM_i}\big]\Big] \right\rVert_1 \nonumber\\
    &= \left\lVert\Tr_{S_r}\Big[\dyad{\Psi}_{S_rE}\prod_{i=1}^{n_r}\Tr_{M_i}\big[(\mathbbm{1}_{S_i} \otimes \dyad{0}_{M_i})C_{S_iM_i}\big]\Big]\right\rVert_1\text{~\ref{set:untrusted_prep}}, \nonumber\\
    \left\lVert \rho_{\mathbf{K}E} - \rho_{\mathbf{K}E}^{\text{ideal}} \right\rVert_1 &= \left\lVert\Tr_{M_r}\Big[\dyad{\Phi}_{M_rE}\prod_{i=1}^{n_r}\Tr_{S_i}\big[(\rho_{S_i} \otimes \mathbbm{1}_{M_i})C_{S_iM_i}\big]\Big]\right\rVert_1\text{~\ref{set:untrusted_meas}}.
\end{align}

The terms inside the trace-norms are operators acting on $E$. Following the variational characterization of the trace-norm, there exist an Hermitian operatos $H$, $\tilde{H}$ acting on $E$ with spectrum $\pm 1$ such that
\begin{align}
    &\left\lVert\Tr_{S_r}\Big[\dyad{\Psi}_{S_rE}\prod_{i=1}^{n_r}\Tr_{M_i}\big[(\mathbbm{1}_{S_i} \otimes \dyad{0}_{M_i})C_{S_iM_i}\big]\Big]\right\rVert_1  \nonumber\\
    &= \Tr_{E}\Bigg[\Tr_{S_r}\Big[\dyad{\Psi}_{S_rE}\prod_{i=1}^{n_r}\Tr_{M_i}\big[(\mathbbm{1}_{S_i} \otimes \dyad{0}_{M_i})C_{S_iM_i}\big]\Big]H_ E\Bigg] \nonumber\\
    &= \Tr_{E}\Bigg[\Tr_{S_r}\Big[\dyad{\Psi}_{S_rE}\prod_{i=1}^{n_r}\Tr_{M_i}\big[(\mathbbm{1}_{S_i} \otimes \dyad{0}_{M_i})C_{S_iM_i}\big]\Big]\big(H_E^+ - H_E^-\big)\Bigg]\text{~\ref{set:untrusted_prep}}, \nonumber\\
    &\left\lVert\Tr_{M_r}\Big[\dyad{\Phi}_{M_rE}\prod_{i=1}^{n_r}\Tr_{S_i}\big[(\rho_{S_i} \otimes \mathbbm{1}_{M_i})C_{S_iM_i}\big]\Big]\right\rVert_1  \nonumber\\
    &= \Tr_{E}\Bigg[\Tr_{M_r}\Big[\dyad{\Phi}_{M_rE}\prod_{i=1}^{n_r}\Tr_{S_i}\big[(\rho_{S_i} \otimes \mathbbm{1}_{M_i})C_{S_iM_i}\big]\Big]\big(\tilde{H}_E^+ - \tilde{H}_E^-\big)\Bigg]\text{~\ref{set:untrusted_meas}},
\end{align}
where $H_E^\pm$ are the projectors corresponding to the spectral decomposition of $H_E$, and similarly for $\tilde{H}_E^\pm$ and $\tilde{H}_E$. Using the formulation of Thm.~\ref{thm: ancilla operator inequality} in Setts.~\ref{set:untrusted_prep} and~\ref{set:untrusted_meas}, and using~\cite[Lemma 6]{Foreman_2025}, it follows that
\begin{align}
    &\pm \left(\prod_{i=1}^{n_r}\Tr_{M_i}\big[(\mathbbm{1}_{S_i} \otimes \dyad{0}_{M_i})C_{S_iM_i}\big]\otimes H_E^{\pm}\right) \preceq \prod_{i=1}^{n_r} G_{S_i} \otimes H_E^\pm\text{~\ref{set:untrusted_prep}}, \nonumber\\
    &\pm \left(\prod_{i=1}^{n_r}\Tr_{S_i}\big[(\rho_{S_i} \otimes \mathbbm{1}_{M_i})C_{S_iM_i}\big]\otimes \tilde{H}_E^{\pm}\right) \preceq \prod_{i=1}^{n_r} G_{M_i} \otimes \tilde{H}_E^\pm\text{~\ref{set:untrusted_meas}},
\end{align}
and therefore
\begin{align}
    \left\lVert \rho_{\mathbf{K}E} - \rho_{\mathbf{K}E}^{\text{ideal}} \right\rVert_1 &\leq \Tr_{E} \Bigg[\Tr_{S_r}\Big[\dyad{\Psi}_{S_rE} \prod_{i=1}^{n_r}G_{S_i} \otimes \big(H_E^+ + H_E^-\big)\Big]\Bigg] \nonumber\\
    &= \Tr_{S_r}\Bigg[\Tr_{E}\Big[\dyad{\Psi}_{S_rE}\Big]\prod_{i=1}^{n_r}G_{S_i}\Bigg] = \Tr\Bigg[\rho_{S_r}\prod_{i=1}^{n_r}G_{S_i}\Bigg]\text{~\ref{set:untrusted_prep}}, \nonumber\\
    \left\lVert \rho_{\mathbf{K}E} - \rho_{\mathbf{K}E}^{\text{ideal}} \right\rVert_1 &\leq \Tr\Bigg[\sigma_{M_r}\prod_{i=1}^{n_r}G_{M_i}\Bigg]\text{~\ref{set:untrusted_meas}},
\end{align}
where we have used that $H_E^+ + H_E^- = \tilde{H}_E^+ + \tilde{H}_E^- = \mathbbm{1}_E$, $\rho_{S_r} = \Tr_E\Big[\dyad{\Psi}_{S_rE}\Big]$, and $\sigma_{M_r} = \Tr_E\Big[\dyad{\Phi}_{M_rE}\Big]$. As mentioned before, the proof for the case in Sett.~\ref{set:SDI} follows directly from the derivations for Sett.~\ref{set:untrusted_meas}.

\section{Proof of Thm.~\ref{thm: multi-bit extraction security proof}}\label{app: proof of multi-bit extraction security proof}

The initial cq-state for an $m-$bit extractor $g:\{0,1\}^{n_r} \rightarrow \{0,1\}^m$ is given by
\begin{equation}\label{eq:cq-state m-bit extractor}
   \rho_{\mathbf{K}E} = \sum_{\mathbf{k}\in\{0,1\}^m}\sum_{\mathbf{a}\in\{0,1\}^{n_r}}\dyad{\mathbf{k}} \otimes \delta_{\mathbf{k},g(\mathbf{a})}\Tr_{S_rM_r}\Bigg[\rho_{S_rM_rE}\prod_{i=1}^{n_r}\Pi_{S_iM_i}^{a_i}\Bigg],
\end{equation}
where $\rho_{S_rM_rE}$ is defined in Eq.~\eqref{eq:initial state} for Setts.~\ref{set:untrusted_prep}-\ref{set:SDI}.

Analogously to the proof of Thm.~\ref{thm: single-bit extraction security proof}, we define the operators $C_{S_iM_i}$. Using Eq.~\eqref{eq:inverse relationship} and expanding the product across rounds, we have that
\begin{equation}\label{eq:expanded product}
   \prod_{i=1}^{n_r} \Pi_{S_iM_i}^{a_i} = 2^{-n_r} \sum_{\mathbf{r}\in\{0,1\}^{n_r}} \prod_{i=1}^{n_r} (-1)^{a_i\: r_i}(C_{S_iM_i})^{r_i},
\end{equation}
noting that $C_{S_iM_i}$ is full-rank. Substituting Eq.~\eqref{eq:expanded product} in Eq.~\eqref{eq:cq-state m-bit extractor}, we obtain
\begin{align}
   \rho_{\mathbf{K}E} &= \sum_{\mathbf{k}}\sum_{\mathbf{a}}\dyad{\mathbf{k}} \otimes \delta_{\mathbf{k},g(\mathbf{a})}\Tr_{S_rM_r}\Bigg[\rho_{S_rM_rE}\: 2^{-n_r}\sum_{\mathbf{r}}\prod_{i=1}^{n_r}(-1)^{a_i\:r_i}(C_{S_iM_i})^{r_i}\Bigg] \nonumber\\
   &= 2^{-n_r}\sum_{\mathbf{k}} \dyad{\mathbf{k}} \otimes \sum_{\mathbf{r}} \Bigg(\sum_{\mathbf{a}} \delta_{\mathbf{k},g(\mathbf{a})}(-1)^{\mathbf{a}\cdot\mathbf{r}}\Bigg)\Tr_{S_rM_r}\Bigg[\rho_{S_rM_rE}\prod_{i=1}^{n_r}(C_{S_iM_i})^{r_i}\Bigg].
\end{align}

The trace distance with the ideal state $\rho_{\mathbf{K}E}^{\text{ideal}}$ is then
\begin{equation}
    \left\lVert \rho_{\mathbf{K}E} - \rho_{\mathbf{K}E}^{\text{ideal}}\right\rVert_1 = 2^{-n_r}\sum_{\mathbf{k}} \left\lVert \sum_{\mathbf{r}} \Bigg(\sum_{\mathbf{a}} \left(\delta_{\mathbf{k},g(\mathbf{a})} - 2^{-m}\right)(-1)^{\mathbf{a}\cdot\mathbf{r}}\Bigg)\Tr_{S_rM_r}\Bigg[\rho_{S_rM_rE}\prod_{i=1}^{n_r}(C_{S_iM_i})^{r_i}\Bigg]\right\rVert_1.
\end{equation}
Using the triangle inequality and the property in Eq.~\eqref{eq:property extractors}, it follows that
\begin{align}
    \left\lVert \rho_{\mathbf{K}E} - \rho_{\mathbf{K}E}^{\text{ideal}}\right\rVert_1 &\leq 2^{-n_r}\sum_{\mathbf{k}}  \sum_{\mathbf{r}} \Bigg|\Bigg(\sum_{\mathbf{a}} \left(\delta_{\mathbf{k},g(\mathbf{a})} - 2^{-m}\right)(-1)^{\mathbf{a}\cdot\mathbf{r}}\Bigg)\Bigg|\quad  \left\lVert\Tr_{S_rM_r}\Bigg[\rho_{S_rM_rE}\prod_{i=1}^{n_r}(C_{S_iM_i})^{r_i}\Bigg]\right\rVert_1\nonumber\\
    &\leq 2^{-n_r} \sum_{\mathbf{k}}\sum_{\mathbf{r}} n_r^2 \sqrt{2}^{n_r - m}\left\lVert\Tr_{S_rM_r}\Bigg[\rho_{S_rM_rE}\prod_{i=1}^{n_r}(C_{S_iM_i})^{r_i}\Bigg]\right\rVert_1\nonumber\\
    &= n_r^2 \sqrt{2}^{m-n_r} \sum_{\mathbf{r}}\left\lVert\Tr_{S_rM_r}\Bigg[\rho_{S_rM_rE}\prod_{i=1}^{n_r}(C_{S_iM_i})^{r_i}\Bigg]\right\rVert_1. \label{eq:expression after triangle ineq}
\end{align}

Following similar steps as in the proof of Thm.~\ref{thm: single-bit extraction security proof}, we substitute $\rho_{S_rM_rE}$ by the expressions in Eq.~\eqref{eq:initial state} for the different settings, introduce the projectors $H^{\pm}_{E}$, and use Thm.~\ref{thm: ancilla operator inequality} and~\cite[Lemma 6]{Foreman_2025} to finally have that
\begin{align}
    \left\lVert \rho_{\mathbf{K}E} - \rho_{\mathbf{K}E}^{\text{ideal}}\right\rVert_1 &\leq n_r^2 \sqrt{2}^{m-n_r} \sum_{\mathbf{r}} \left(\Tr_{S_rE}\Bigg[\dyad{\Psi}_{S_rE}\prod_{i=1}^{n_r}\Tr_{M_i}\Big[\left(\mathbbm{1}_{S_i} \otimes \dyad{0}_{M_i}\right)(C_{S_iM_i})^{r_i}\Big]\otimes (H_E^+-H_E^-)\Bigg]\right) \nonumber\\
    &\leq n_r^2 \sqrt{2}^{m-n_r} \sum_{\mathbf{r}} \Tr\Bigg[\rho_{S_r}\prod_{i=1}^{n_r}(G_{S_i})^{r_i}\Bigg] = n_r^2 \sqrt{2}^{m-n_r} \Tr\Bigg[\rho_{S_r}\prod_{i=1}^{n_r}(\mathbbm{1}_{S_i} + G_{S_i})\Bigg],
\end{align}
where we only show the derivation for Sett.~\ref{set:untrusted_prep}, but the results for Sett.~\ref{set:untrusted_meas} and~\ref{set:SDI} are obtained analogously.

\section{Proof of Thm.~\ref{thm:spot-checking protocol security}}\label{app:theorem spot-checking}

In each estimation round, the operator $Q_{M_j}^{a_j,x_j}$ is applied to the internal system of the measurement device, given by
\begin{equation}
    Q_{M_j}^{a_j,x_j} = \frac{1}{2} \Tr_{S_j}\Bigg[\left(\rho_{S_j}^{x_j} \otimes \mathbbm{1}_{M_j}\right) \Pi_{S_jM_j}^{a_j}\Bigg],
\end{equation}
and it can be verified that
\begin{equation}\label{eq:identities Q}
    \sum\limits_{a_j,x_j=0}^1 Q_{M_j}^{a_j,x_j} = \mathbbm{1}_{M_j}, \qquad
     G_{M_i} = \sum\limits_{a_i,x_i=0}^1 (4\nu_{a_i,x_i} - 1)Q_{M_i}^{a_i,x_i}.
\end{equation}

The labels $\mathbf{t}$ can be assumed to be all generated before starting the protocol, so the rounds can be reordered without loss of generality. The probability of obtaining estimation data $\mathbf{z}$ given labels $\mathbf{t}$ is
\begin{equation}
    p(\mathbf{z}|\mathbf{t}) = \Tr\Bigg[\sigma_{M_e | \mathbf{t}}\prod_{j=1}^{n_e}Q_{M_j}^{a_j,x_j}\Bigg].
\end{equation}

The internal state across raw key generation rounds conditioned on a particular estimation data $\mathbf{z}$ is given by
\begin{equation}\label{eq:raw key state}
    \sigma_{M_r|\mathbf{t},\mathbf{z}} = \frac{1}{p(\mathbf{z}|\mathbf{t})} \Tr_{M_e}\Bigg[\sigma_{M_eM_r|\mathbf{t}}\prod_{j=1}^{n_e}Q_{M_j}^{a_j,x_j}\Bigg].
\end{equation}

Next, we use the results from Sec.~\ref{sec:extractors} to bound the trace distance for any choice of $\mathbf{t}$ and $\mathbf{z}$. First, for the \texttt{XOR} extractor, from Thm.~\ref{thm: single-bit extraction security proof} we have that
\begin{equation}
    \left\lVert \rho_{\mathbf{K}E|\mathbf{t},\mathbf{z}} - \rho_{\mathbf{K}E|\mathbf{t},\mathbf{z}}^{\text{ideal}} \right\rVert_1 \leq \Tr\Bigg[\sigma_{M_r|\mathbf{t},\mathbf{z}}\prod_{i=1}^{n_r}G_{M_i}\Bigg].
\end{equation}
Inserting Eq.~\eqref{eq:raw key state} into the above inequality, we obtain
\begin{equation}
   \left\lVert \rho_{\mathbf{K}E|\mathbf{t},\mathbf{z}} - \rho_{\mathbf{K}E|\mathbf{t},\mathbf{z}}^{\text{ideal}} \right\rVert_1 \leq \frac{1}{p(\mathbf{z}|\mathbf{t})}\Tr\Bigg[\sigma_{M_eM_r|\mathbf{t}}\prod_{j=1}^{n_e}Q_{M_j}^{a_j,x_j}\prod_{i=1}^{n_r}G_{M_i}\Bigg],
\end{equation}
and using that $p(\mathbf{t},\mathbf{z}) = p(\mathbf{z}|\mathbf{t})p(\mathbf{t})$ and $p(\mathbf{t}) =p_e^{n_e} p_r^{n_r}$, it follows that
\begin{equation}
    \sum_{\mathbf{t},\mathbf{z}} p(\mathbf{t},\mathbf{z})\left\lVert \rho_{\mathbf{K}E|\mathbf{t},\mathbf{z}} - \rho_{\mathbf{K}E|\mathbf{t},\mathbf{z}}^{\text{ideal}} \right\rVert_1 \leq \sum_{\mathbf{t},\mathbf{z}} \Tr\Bigg[\sigma_{M_eM_r|\mathbf{t}}\prod_{j=1}^{n_e}\Big(p_eQ_{M_j}^{a_j,x_j}\Big)\prod_{i=1}^{n_r}\Big(p_rG_{M_i}\Big)\Bigg].
\end{equation}

Note that when $m(\mathbf{t},\mathbf{z}) = 0$, no output state is produced, and the trace distance is trivially $0$. Therefore, the inequality above can be rewritten as
\begin{align}
    &\sum_{\mathbf{t},\mathbf{z}} p(\mathbf{t},\mathbf{z})\left\lVert \rho_{\mathbf{K}E|\mathbf{t},\mathbf{z}} - \rho_{\mathbf{K}E|\mathbf{t},\mathbf{z}}^{\text{ideal}} \right\rVert_1 \leq \sum_{\mathbf{t},\mathbf{z}} m(\mathbf{t},\mathbf{z})\Tr\Bigg[\sigma_{M_eM_r|\mathbf{t}}\prod_{j=1}^{n_e}\Big(p_eQ_{M_j}^{a_j,x_j}\Big)\prod_{i=1}^{n_r}\Big(p_rG_{M_i}\Big)\Bigg] \nonumber \\
    &\leq \sum_{\mathbf{t},\mathbf{z}} \sqrt{2}^{\sum_{j=1}^{n_e} \alpha_{a_j,x_j} + (\beta - 1)n_r - 2\log(1/\epsilon)}\Tr\Bigg[\sigma_{M_eM_r|\mathbf{t}}\prod_{j=1}^{n_e}\Big(p_eQ_{M_j}^{a_j,x_j}\Big)\prod_{i=1}^{n_r}\Big(p_rG_{M_i}\Big)\Bigg], \label{eq:bound on error}
\end{align}
where $\{\alpha_{a,x}\}_{a,x=0}^1$ and $\beta$ are the optimal values that maximize Eq.~\eqref{eq:length single-bit extractor}, subject to the constraints in Eq.~\eqref{eq:constraints single-bit extractor length}. The bound in Eq.~\eqref{eq:bound on error} follows from the fact that when
$m(\mathbf{t},\mathbf{z}) = 1$, we have
\begin{equation}
    m(\mathbf{t},\mathbf{z}) \leq \sqrt{2}^{\sum_{j=1}^{n_e} \alpha_{a_j,x_j} + (\beta - 1)n_r - 2\log(1/\epsilon)}.
\end{equation}

Performing the summation over $\mathbf{z}$, Eq.~\eqref{eq:bound on error} becomes
\begin{equation}
    \epsilon\sum_{\mathbf{t}} \Tr\Bigg[\sigma_{M_eM_r|\mathbf{t}}\prod_{j=1}^{n_e}\Big(p_e\sum_{a_j,x_j=0}^1 \sqrt{2}^{\alpha_{a_j,x_j}}Q_{M_j}^{a_j,x_j}\Big)\prod_{i=1}^{n_r}\Big(p_r\sqrt{2}^{\beta - 1} G_{M_i}\Big)\Bigg].
\end{equation}
Using the second identity in Eq.~\eqref{eq:identities Q}, we can write
\begin{equation}
    \epsilon\sum_{\mathbf{t}} \Tr\Bigg[\sigma_{M_eM_r|\mathbf{t}}\prod_{j=1}^{n_e}\Big(p_e\sum_{a_j,x_j=0}^1 \sqrt{2}^{\alpha_{a_j,x_j}}Q_{M_j}^{a_j,x_j}\Big)\prod_{i=1}^{n_r}\Big(p_r\sqrt{2}^{\beta - 1} \sum\limits_{a_i,x_i=0}^1 (4\nu_{a_i,x_i} - 1)Q_{M_i}^{a_i,x_i}\Big)\Bigg].
\end{equation}
Performing the summation over $\mathbf{t}$, it simplifies to
\begin{equation}\label{eq:expanded term}
\epsilon\Tr\Bigg[\sigma_{M_g}\prod_{k=1}^{n}\sum_{a_k,x_k=0}^1 \Big(p_e\sqrt{2}^{\alpha_{a_k,x_k}} + p_r\sqrt{2}^{\beta - 1}(4\nu_{a_k,x_k} - 1)\Big)Q_{M_k}^{a_k,x_k}\Bigg].
\end{equation}
Using the constraints in Eq.~\eqref{eq:constraints single-bit extractor length} and the first identity in Eq.~\eqref{eq:identities Q}, we have that
\begin{equation}
    \sum_{a_k,x_k=0}^1 \Big(p_e\sqrt{2}^{\alpha_{a_k,x_k}} + p_r\sqrt{2}^{\beta - 1}(4\nu_{a_k,x_k} - 1)\Big)Q_{M_k}^{a_k,x_k} = \sum_{a_k,x_k=0}^1 Q_{M_k}^{a_k,x_k} = \mathbbm{1}_{M_k}.
\end{equation}
Inserting this into Eq.~\eqref{eq:bound on error} and using that $\Tr[\sigma_{M_g}] = 1$, we finally have
\begin{equation}
    \sum_{\mathbf{t},\mathbf{z}} p(\mathbf{t},\mathbf{z})\left\lVert \rho_{\mathbf{K}E|\mathbf{t},\mathbf{z}} - \rho_{\mathbf{K}E|\mathbf{t},\mathbf{z}}^{\text{ideal}} \right\rVert_1 \leq \epsilon.
\end{equation}

In the case of the $m$-bit extractors, we introduce an indicator function in order to encode the case where no output is produced, given by
\begin{equation}
    I(m(\mathbf{t},\mathbf{z})) = \begin{cases}
        1 \quad \text{if} \quad m(\mathbf{t},\mathbf{z}) > 0,\\
        0 \quad \text{otherwise}.
    \end{cases}
\end{equation}
Using Thm.~\ref{thm: multi-bit extraction security proof}, we have that
\begin{align}
    &\sum_{\mathbf{t}}\sum_{\mathbf{z}} p(\mathbf{t},\mathbf{z}) \left\lVert \rho_{\mathbf{K}E|\mathbf{t},\mathbf{z}} - \rho_{\mathbf{K}E|\mathbf{t},\mathbf{z}}^{\text{ideal}} \right\rVert_1 \nonumber \\
    &\leq \sum_{\mathbf{t},\mathbf{z}} p(\mathbf{t},\mathbf{z}) I(m(\mathbf{t},\mathbf{z})) \sqrt{2}^{m(\mathbf{t},\mathbf{z}) - n_r + 4\log{n_r}} \Tr\Bigg[\sigma_{M_r|\mathbf{t},\mathbf{z}} \prod_{i=1}^{n_r}\left(\mathbbm{1}_{M_i} + G_{M_i}\right)\Bigg] \nonumber \\
    &\leq \sum_{\mathbf{t},\mathbf{z}} p(\mathbf{t},\mathbf{z}) \sqrt{2}^{m(\mathbf{t},\mathbf{z}) - n_r + 4\log{n_r}} \Tr\Bigg[\sigma_{M_r|\mathbf{t},\mathbf{z}} \prod_{i=1}^{n_r}\left(\mathbbm{1}_{M_i} + G_{M_i}\right)\Bigg]. \label{eq:mbit extractor theorem 4}
\end{align}

Substituting Eq.~\eqref{eq:length m-bit extractors} for the output length, Eq.~\eqref{eq:raw key state} for the internal state of the measurement device conditioned on the estimation data, and $p(\mathbf{t},\mathbf{z}) = p(\mathbf{z}|\mathbf{t}) p_e^{n_e}p_r^{n_r}$, the expression above can be written as
\begin{equation}
    \epsilon \sum_{\mathbf{t},\mathbf{z}} p_e^{n_e}p_r^{n_r} \sqrt{2}^{\sum_{j=1}^{n_e} \alpha_{a_j,x_j} + (\beta-1)n_r} \Tr\Bigg[\sigma_{M_eM_r|\mathbf{t}} \prod_{j=1}^{n_e}Q_{M_j}^{a_j,x_j}\prod_{i=1}^{n_r}\left(\mathbbm{1}_{M_i} + G_{M_i}\right)\Bigg],
\end{equation}
which, when performing the summation over $\mathbf{z}$, becomes
\begin{equation}
    \epsilon \sum_{\mathbf{t}} \Tr\Bigg[\sigma_{M_eM_r|\mathbf{t}} \prod_{j=1}^{n_e}\Big(p_e\sum_{a_j,x_j=0}^1 \sqrt{2}^{\alpha_{a_j,x_j}}Q_{M_j}^{a_j,x_j}\Big)\prod_{i=1}^{n_r} \Big(p_r\sqrt{2}^{\beta-1}(\mathbbm{1}_{M_i} + G_{M_i})\Big)\Bigg].
\end{equation}
Using the second identity in Eq.~\eqref{eq:identities Q} and performing the summation over $\mathbf{t}$, it reduces to
\begin{equation}
    \epsilon \Tr\Bigg[\sigma_{M_g} \prod_{k=1}^{n}\sum_{a_k,x_k=0}^1\Big(p_e \sqrt{2}^{\alpha_{a_k,x_k}} + 4p_r\sqrt{2}^{\beta-1}\nu_{a_k,x_k}\Big)Q_{M_k}^{a_k,x_k}\Bigg].
\end{equation}
Using the constraints in Eq.~\eqref{eq:constraints m-bit extractor length} and the first identity in Eq.~\eqref{eq:identities Q}, we have that
\begin{equation}
    \sum_{a_k,x_k=0}^1 \Big(p_e\sqrt{2}^{\alpha_{a_k,x_k}} + 4p_r\sqrt{2}^{\beta - 1}\nu_{a_k,x_k}\Big)Q_{M_k}^{a_k,x_k} = \sum_{a_k,x_k=0}^1 Q_{M_k}^{a_k,x_k} = \mathbbm{1}_{M_k}.
\end{equation}
Inserting this into Eq.~\eqref{eq:mbit extractor theorem 4}, we finally have
\begin{equation}
    \sum_{\mathbf{t},\mathbf{z}} p(\mathbf{t},\mathbf{z})\left\lVert \rho_{\mathbf{K}E|\mathbf{t},\mathbf{z}} - \rho_{\mathbf{K}E|\mathbf{t},\mathbf{z}}^{\text{ideal}} \right\rVert_1 \leq \epsilon.
\end{equation}

\end{document}